\documentclass[aps,prl,10pt,floatfix,superscriptaddress,longbibliography,showpacs,twocolumn]{revtex4-2}

\newif\ifapp

\usepackage{silence}
\usepackage{float}
\usepackage{makecell}
\usepackage{hyperref}
\usepackage{physics}
\usepackage{xfrac}
\usepackage{graphicx}
\usepackage{subfigure}
\usepackage{mathdots}
\usepackage{amsfonts,amssymb,amsmath}
\usepackage{mathrsfs}
\usepackage[]{graphics,graphicx,epsfig}
\usepackage{amsthm}
\usepackage{csquotes}
\MakeOuterQuote{"}
\usepackage{epstopdf}
\usepackage{tikz}
\usepackage{booktabs}

\usepackage{algorithm}
\usepackage{algpseudocode}
\usepackage{silence}
\usepackage{paralist}
\usepackage{diagbox}
\usepackage[inline]{enumitem}

\usepackage{tabularx}
\usepackage{setspace}
\usepackage{mathtools}
\usepackage{dsfont}

\usepackage{mathtools,mleftright}
\mleftright
\hypersetup{
    breaklinks=true,
    colorlinks=true,       
    linkcolor=blue,          
    citecolor=red,        
    filecolor=magenta,      
    urlcolor=blue,           
    runcolor=cyan
}

\newcommand{\bbE}{\mathbb{E}}
\newcommand{\bbN}{\mathbb{N}}
\newcommand{\bbR}{\mathbb{R}}
\newcommand{\bbS}{\mathbb{S}}

\newcommand{\bbF}{\mathbb{F}}

\newcommand{\caA}{\mathcal{A}}
\newcommand{\caB}{\mathcal{B}}

\newcommand{\caE}{\mathcal{E}}
\newcommand{\caH}{\mathcal{H}}

\newcommand{\caK}{\mathcal{K}}
\newcommand{\caL}{\mathcal{L}}
\newcommand{\caM}{\mathcal{M}}
\newcommand{\caN}{\mathcal{N}}
\newcommand{\caP}{\mathcal{P}}
\newcommand{\caR}{\mathcal{R}}

\newcommand{\caU}{\mathcal{U}}
\newcommand{\caW}{\mathcal{W}}
\newcommand{\caX}{\mathcal{X}}

\newcommand{\caZ}{\mathcal{Z}}
\newcommand{\rmA}{\mathrm{A}}
\newcommand{\rmB}{\mathrm{B}}
\newcommand{\rmH}{\mathrm{H}}

\newcommand{\rmi}{\mathrm{i}}

\newcommand{\sn}{\mathrm{SN}}

\newcommand{\Haar}{\mathrm{Haar}}

\newcommand{\poly}{\mathrm{poly}}
\def\<{\langle}  
\def\>{\rangle}  

\newtheorem{theorem}{Theorem}

\newtheorem{lemma}{Lemma}
\newtheorem{proposition}{Proposition}

\newtheorem{corollary}{Corollary}

\newcommand{\lref}[1]{Lemma~\ref{#1}}

\newcommand{\thref}[1]{Theorem~\ref{#1}}

\newcommand{\pref}[1]{Proposition~\ref{#1}}

\newcommand{\coref}[1]{Corollary~\ref{#1}}

\newcommand{\eref}[1]{Eq.~\textup{(\ref{#1})}}

\newcommand{\tref}[1]{Table~\ref{#1}}

\newcommand{\fref}[1]{Fig.~\ref{#1}}

\newcommand{\aref}[1]{Appendix~\ref{#1}}

\newcommand{\bbC}{\mathbb{C}}

\theoremstyle{plain} 
\theoremstyle{plain} 
\theoremstyle{plain} 
\theoremstyle{plain}

\newcommand{\bfG}{\mathbf{G}}
\newcommand{\caV}{\mathcal{V}}

\newcommand{\caQ}{\mathcal{Q}}
\newcommand{\lb}{\mathrm{lb}}
\newcommand{\bbU}{\mathbb{U}}
\newcommand{\bfP}{\mathbf{P}}

\begin{document}
    \title{Certifying unknown-basis pure-state entanglement dimensionality with few measurement settings}
    \author{Changhao Yi}
    \affiliation{Department of Physics, Shanghai University, Shanghai 200444, China}
    \date{\today}
    
\begin{abstract}
      We introduce a basis-independent algorithm that only uses single-copy measurements to certify the Schmidt number 
      of $2n$-qubit bipartite pure states. Assuming the local system dimension is $d = 2^n$, the Schmidt number $\chi$ is $O(1)$, and the non-zero Schmidt coefficients are bounded below by $\Omega(1)$, our protocol achieves a total sample complexity of $\widetilde{O}(d^2\chi^2)$ 
      with only $\widetilde{O}(\chi^4)$ measurement settings. This work establishes a measurement-setting-efficient approach to high-dimensional entanglement 
      certification.
\end{abstract} 

\maketitle

\textit{Introduction.}---
Quantum entanglement is a hallmark of quantum mechanics and a central
resource for quantum information processing. Its characterization has therefore
been a longstanding objective of quantum information science
\cite{guhne2009entanglement}. Beyond separability, the \textit{Schmidt number}, or
\textit{entanglement dimensionality}, quantifies the number of local degrees of
freedom participating in entanglement \cite{terhal2000schmidt}.
High-dimensional entanglement offers advantages in quantum communication
\cite{zhang2025quantum}, channel discrimination \cite{bae2019more}, quantum key
distribution \cite{huber2013weak}, and linear-optical fusion
\cite{yamazaki2025linear}. Exploiting these advantages requires scalable
certification of experimentally prepared states, which is inaccessible to full
state tomography because of its exponential cost.

A broad range of tomography-free protocols has been developed, including
fidelity-based witnesses
\cite{huang2016high,bavaresco2018measurements,morelli2023resource,li2025high},
correlation-matrix methods
\cite{liu2023characterizing,wyderka2023probing,lib2024experimental},
$k$-positive maps \cite{yi2025certifying,mallick2025detecting}, verification
operators \cite{wu2025efficient}, nonlocal games \cite{mukh2025measure}, and
semidefinite-programming relaxations
\cite{hu2021optimized,dalessandro2025sdp}. These approaches exhibit a basic
trade-off. \textit{Basis-dependent} methods, such as fidelity witnesses, use only a
constant number of measurement settings, but require prior knowledge of the
Schmidt basis and can lose certification power under local basis
misalignment. \textit{Basis-independent} methods, including the correlation matrix
protocols, avoid this assumption but usually require a number of
measurement settings that grows at least linearly with the local dimension. As a trade-off, the methods usually
employ high-order unitary designs
\cite{zhu2017multiqubit,gross2007evenly}, whose experimental realization is
demanding. Here, basis independence refers specifically to independence from
the unknown Schmidt basis, rather than from the chosen laboratory measurement
basis.

In strongly correlated many-body systems the Schmidt vectors are generally complex many-body states with no simple relation to a fixed product basis. Determining them can be comparable in
difficulty to classically simulating the state. Analogous problems arise in
photonic platforms, where spatial or spectral Schmidt modes depend sensitively
on source parameters and experimental drift. In these situations, basis-dependent
methods can significantly underestimate the entanglement dimensionality.
This motivates certification protocols that are simultaneously
basis-independent, measurement-efficient, and experimentally accessible.

In this work, we develop a basis-independent protocol for certifying the Schmidt number of bipartite pure states using a single pair of local random unitaries. A comparison with other methods is summarized in \tref{table:cost}. Our starting point is that the rank of the correlation matrix equals the square of the Schmidt number. In the worst case (such as the computational basis scenario), preserving this rank requires $\Omega(d)$ Pauli operators, corresponding to $\Omega(d^2)$ measurement settings. After local random scrambling, only $\widetilde{O}(\chi^2)$ Pauli operators, or $\widetilde{O}(\chi^4)$ measurement settings, suffice with high probability. Throughout, $\widetilde{O}$ suppresses factors polylogarithmic in $d$, $\chi$, and $\eta^{-1}$. For $\chi = O(1)$, this gives an exponential reduction in the number of settings as a function of the number of qubits compared with low-rank tomography. The protocol primarily reduces measurement-setting complexity rather than the total number of state copies, which is advantageous when changing measurement configurations is more costly than repeating a fixed setting. We highlight random local unitaries \cite{ji2018pseudorandom,Brandao2016LocalRandom,schuster2025random} as a practical tool for basis-independent certification of high-dimensional entanglement.

\renewcommand{\arraystretch}{1.4}
\begin{table*}[t]
\centering 
\caption{\label{table:cost} Comparison of the experimental costs of different methods. The relevant parameters are: subsystem dimension $d$, Schmidt number $\chi$, and the lower bound on non-zero Schmidt coefficients $\lambda_{\rm lb}$. The number of measurement settings is the number of distinct measurement configurations that the experimental apparatus must implement. The number of shots per setting indicates how many times that configuration is repeated. The total sample complexity is the total number of state copies consumed, which is the product of the previous two quantities.}
\begin{tabular}{c|c|c|c|c}
\hline
Method & Measurement settings & Shots per setting & Total samples & Basis-independent  \\
\hline
Tomography\,\cite{gross2010quantum,liu2011universal} & $\widetilde{O}(d\chi)$ & $O(d\lambda_{\rm lb}^{-2})$ & $\widetilde{O}(d^2\chi\lambda^{-2}_{\rm lb})$ & Yes \\
Fidelity-based\,\cite{bavaresco2018measurements} & $O(1)$ & $O(d^2)$ & $O(d^2)$ & No   \\
Projected CM & $\widetilde{O}(\chi^4)$ & $\widetilde{O}(d^2 \chi^{-2} \lambda_{\rm lb}^{-2})$ & $\widetilde{O}(d^2\chi^2\lambda_{\rm lb}^{-2})$ & Yes   \\
\hline
\end{tabular}
\end{table*}

\textit{Schmidt-number certification.}---Let $\sn(\rho)$ denote the Schmidt number of a bipartite state $\rho$ \cite{terhal2000schmidt}. Consider a bipartite pure state \(|\psi\rangle\) supported on \(\caH_{\rmA\rmB} = \mathcal{H}_\rmA \otimes \mathcal{H}_\rmB\) with \(\dim(\mathcal{H}_\rmA) = \dim(\mathcal{H}_\rmB) = d\) and \(\sn(|\psi\>) = \chi\):
\begin{equation}
|\psi\rangle = \sum_{i=1}^{\chi} \sqrt{\lambda_i} \, |l_i\rangle \otimes |r_i\rangle, \quad \lambda_i > 0,
\end{equation}
where \(\{\lambda_i\}_{i=1}^{\chi}\) is the set of non-zero \textit{Schmidt coefficients}, and \(\{|l_i\rangle\}_{i=1}^{\chi}\), \(\{|r_i\rangle\}_{i=1}^{\chi}\) are the left/right \textit{Schmidt vectors}. We further assume that $\lambda_i$ has a constant lower bound $\lambda_{\rm lb} = \Omega(1)$, so that $\chi \le 1/\lambda_{\rm lb} = O(1)$. Throughout this work we focus on the regime of constant Schmidt number, while displaying the $\chi$-dependence of all bounds explicitly.

Let \(\caP = \{W_a\}_{a=0}^{d^2-1}\) be the set of Pauli strings, and $\caP_0 \coloneqq \caP\setminus\{I\}$. We define the full \textit{correlation matrix} (CM) of the state as
\begin{gather}
    \bfG_\psi \coloneqq \sum_{a,b=0}^{d^2-1}\<\psi|(W_a \otimes W_b)|\psi\> |a\rangle\langle b|, 
\end{gather}
where \(\{|a\rangle\}_{a=0}^{d^2-1}\) is the computational basis of \(\mathbb{C}^{d^2}\). A key property is the following. See Supplemental Material \cite{SM} Sec.~III for the proof.
\begin{lemma}\label{lem: Tdecomposition}
    Suppose \(|\psi\rangle = \sum_{i=1}^{\chi}\sqrt{\lambda_i}|l_i\>\otimes |r_i\>\) is a pure state with Schmidt number \(\chi\). Then \(\rank(\bfG_\psi) = \chi^2\) and its singular values are $\{d\sqrt{\lambda_i \lambda_j}\}_{i,j=1}^{\chi}$.
\end{lemma}

The projected CM is constructed as follows. Let $a_0 \coloneqq 0\in\bbF_2^{2n}$ and $ W_0 \coloneqq I$. Generate a sequence of binary strings $\caA \coloneqq (a_1,a_2,\ldots,a_S)$ with each $a_s$ i.i.d. uniformly sampled from $\bbF_2^{2n}\setminus\{0\}$. Note that repetition is allowed. Then we construct the projected CM $\bfG_{\psi,\caA}$ according to
\begin{equation} \label{eq:projectedCM}
\bfG_{\psi,\caA} \coloneqq \sum_{s,t=0}^S \<\psi|(W_{a_s} \otimes W_{a_t})|\psi\> |s\rangle\langle t|.
\end{equation}

Equivalently, $\bfG_{\psi,\caA}$ is obtained from the full CM by
selecting the rows and columns indexed by $(a_0,a_1,\ldots,a_S)$, with
multiplicities retained. Therefore, we have 
\begin{equation}
\rank(\bfG_{\psi,\caA}) \leq \rank(\bfG_{\psi}) = \chi^2.
\end{equation}
Consequently, for any integer $r\ge 1$,  $\rank(\bfG_{\psi,\caA}) > r^2$ certifies $\sn(|\psi\>) \ge r+1$.

\textit{Randomized preprocessing.}---An immediate question is how large $S$ should be to preserve the rank, i.e., keep $\operatorname{rank}(\bfG_\psi) = \operatorname{rank}(\bfG_{\psi,\caA})$. We find that the answer is directly related to the \textit{incoherence parameter}, which measures the maximum overlap of any Pauli operator with the Schmidt basis: 
 \begin{gather}
     \mu(U) \coloneqq d\max_{W_a\in\caP_0}\sum_{i,j=1}^{\chi}|\<i|U^\dag W_a U|j\>|^2.
 \end{gather}
 Here $\{|i\>\}_{i=1}^{\chi}$ are elements of the computational basis (Pauli $Z$ basis).

Introduce unitary operators $T_l,T_r$ such that $|l_i\> = T_l|i\>, |r_i\> = T_r|i\>$ for $i = 1,\ldots,\chi$. Then the target state is
 \begin{equation}
     |\psi\> = \sum_{i=1}^{\chi}\sqrt{\lambda_i}T_l|i\> \otimes T_r|i\>.
 \end{equation}
 In Sec.~IV of \cite{SM}, we prove that 
\begin{equation}
S = \widetilde{O}\left(\max\left\{\mu(T_l), \mu(T_r)\right\}\right) 
\end{equation}
Pauli strings are sufficient to keep $\operatorname{rank}(\bfG_{\psi,\caA}) = \operatorname{rank}(\bfG_{\psi})$ with constant probability. In a worst-case scenario (e.g., $T_l = T_r = I$), we have $\mu(I) = \Theta(d\chi)$, meaning $S$ should grow linearly with the subsystem dimension. However, note that the Schmidt number is invariant under local unitary transformations:
\begin{equation}
    |\psi_U\> \coloneqq U_\rmA \otimes U_\rmB|\psi\>,\quad \sn(|\psi\>) = \sn(|\psi_U\>).
\end{equation}
We term the \textit{randomized preprocessing} subroutine as sampling $U_\rmA, U_\rmB$ from the Haar random unitary ensemble and applying $U_\rmA \otimes U_\rmB$ to the target state before measurements. Note that we only need to sample $U_\rmA \otimes U_\rmB$ once. Hence, the requirement on $S$ is determined by $\max\{\mu(U_\rmA T_l), \mu(U_\rmB T_r)\}$ instead. Because $U_\rmA T_l, U_\rmB T_r$ also follow the statistics of Haar random unitaries, both $\mu(U_{\rmA}T_l) $ and $\mu(U_{\rmB}T_r)$ drop to $O(\log(d)\chi^2)$, leading to an exponential reduction from $S = O(d)$ to $S = O(\log d)$. 

For the maximally entangled state where $\chi = d$, the incoherence parameter equals $d^2$ for all $U$, so the randomized preprocessing cannot reduce the measurement-setting complexity. For the product state where $\chi = 1$, the exponential reduction in $\mu(U)$ holds. However, this situation is trivial as $\operatorname{rank}(\bfG_{\psi_U,\caA}) = 1$ for all $U$ and $S$.

\begin{figure}
    \centering
    \includegraphics[width=\linewidth]{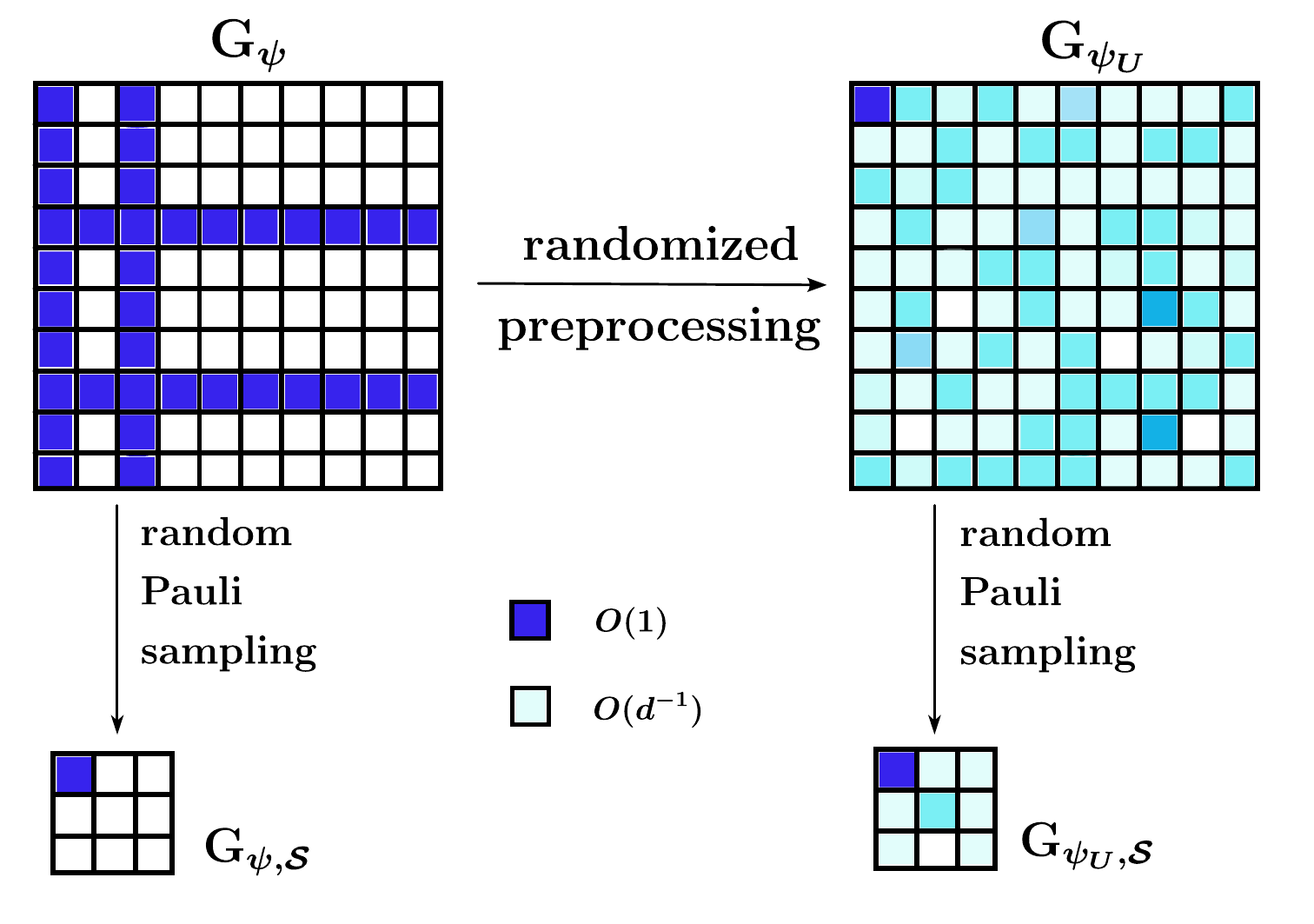}
    \caption{A schematic demonstration of the influence of randomized preprocessing on the CM. We use the color intensity to reflect the magnitudes of matrix entries. For a state $|\psi\>$ with large incoherence parameters, its CM is too sparse, hence we cannot preserve its rank with few Pauli string samples. The randomized preprocessing $|\psi\> \to U_\rmA \otimes U_\rmB|\psi\>$ spreads the matrix entries, so the CM becomes dense and a small projected CM can preserve the rank. The cost is that the entries become exponentially small, and we need exponentially many measurements to estimate each of them.}
    \label{fig:pcm}
\end{figure}

\textit{Main theorem.}---We present our main theorem here, which shows that with the randomized preprocessing, $\widetilde{O}(\chi^2)$ Pauli strings suffice to keep $\operatorname{rank}(\bfG_{\psi_U,\caA}) = \operatorname{rank}(\bfG_{\psi_U})$. The proof is in Sec.~IV of \cite{SM}.

\begin{theorem}\label{theorem:main_formal}
    Suppose $|\psi\> = \sum_{i=1}^{\chi}\sqrt{\lambda_i}|l_i\>\otimes |r_i\>$ is a rank-$\chi$ pure state in $\caH_\rmA \otimes \caH_\rmB$ and $d_\rmA = d_\rmB = d$. Let $\caA = (a_s)_{s=1}^S$ be a sequence of random binary strings with each $a_s$ i.i.d. uniformly drawn from $\bbF^{2n}_2 \setminus \{0\}$. Sample two Haar random unitaries $U_{\rmA}, U_{\rmB}$ on $\caH_\rmA, \caH_\rmB$ separately. Let $\bfG_{\psi_U, \caA}$ be the projected CM of $|\psi_U\> \coloneqq (U_{\rmA} \otimes U_{\rmB})|\psi\>$. Fix $0 < \kappa,\eta < 1$. Then 
\begin{equation}
        S = O\left(\kappa^{-2}\chi^2 \log(d\chi^2\eta^{-1})\log(\chi^2\eta^{-1})\right)
    \end{equation}
    is sufficient to ensure
    \begin{equation}
        \sigma_{\chi^2}(\bfG_{\psi_U,\caA}) \ge (1-\kappa)\frac{dS}{d^2-1}\lambda_{\rm lb},
    \end{equation}
    with failure probability at most $\eta$.
\end{theorem}

The next theorem demonstrates why the randomized preprocessing is necessary, which is proved in Sec.~V of \cite{SM}. 

\begin{theorem}
    Suppose $\chi\ge 2$. Let $\{|l_i\>\}_{i=1}^{\chi}, \{|r_i\>\}_{i=1}^{\chi}$ be elements of the computational basis. In order to have 
    $\mathrm{rank}\left(\bfG_{\psi,\caA}\right) = \chi^2$ with constant probability, we need at least $S = \Omega(d)$.
\end{theorem}

Intuitively, without the randomized preprocessing $U_{\rmA} \otimes U_{\rmB}$, the matrix $\bfG_\psi$ is sparse, and the non-zero entries are large; hence when we perform random Pauli sampling, most entries of $\bfG_\psi$ are 0, and the non-zero entries are $O(1)$. After the randomized preprocessing, the matrix $\bfG_{\psi_U}$ becomes dense, while most entries are $O(d^{-1})$. The full demonstration is in \fref{fig:pcm}.

In practice, we cannot directly sample from the Haar random unitary ensemble. Instead, we sample from a high-order approximate unitary design \cite{Brandao2016LocalRandom,schuster2025random}. The next theorem specifies the required design order and circuit depth. See Sec.~VI of \cite{SM} for the proof. 
\begin{theorem}
\label{thm:design-incoherence}
Suppose $U$ is drawn from an
$\varepsilon$-approximate unitary $2q$-design on $\mathbb{U}(d)$ with
\begin{gather}
  q = \Theta \left(\log(d\chi\eta^{-1})\right),
  \quad
  \varepsilon = O(d^{-q-2}q^q).
\end{gather}
Then $\mu(U) = O(\chi^{2}\log (d\chi\eta^{-1}))$ with probability at least $1-\eta$. The circuit depth for constructing the random unitary is 
\begin{equation}
    O\left(\log (d) \log^2 \left(d\chi\eta^{-1}\right) \operatorname{polylog} \! \log \left(d\chi\eta^{-1}\right)\right).
\end{equation}
\end{theorem}

\textit{Measurement-noise-induced threshold.}---Here we further account for statistical noise due to finite sampling. Note that the original CM has dimension $d^2 \times d^2$, while the projected CM has dimension $(S+1) \times (S+1)$. Let $M$ be the number of measurement shots for each Pauli operator $W_a \otimes W_b$. Denote the measurement-induced error by $\delta \bfG_{\psi,\caA}$, so that the experimentally obtained CM is $\hat{\bfG}_{\psi,\caA} \coloneqq \bfG_{\psi,\caA} + \delta \bfG_{\psi,\caA}$, where each entry of $ \delta \bfG_{\psi,\caA} $ is a sub-Gaussian random variable with standard deviation bounded by $1/\sqrt{M}$. Standard results on random matrices \cite{tropp2015introduction} give $\|\delta\bfG_{\psi,\caA}\| \lesssim \sqrt{S/M}$
with high probability. In order to witness $\chi^2$ singular values, it suffices to set up a threshold value $\sigma_{\rm th}(\eta)$ such that with probability at least $1-\eta$,
\begin{equation}
     \sigma_{\chi^2}(\hat{\bfG}_{\psi,\caA}) \ge \sigma_{\rm th}(\eta) \ge \sigma_{\chi^2+1}(\hat{\bfG}_{\psi,\caA}).
\end{equation}
Combining \thref{theorem:main_formal} and the noise-induced threshold, we design a practical Schmidt number certification protocol (Algorithm~\ref{alg:rank_certification}) whose performance is guaranteed by the following corollary. The proof is in Sec.~VII A of \cite{SM}. 

\begin{corollary}\label{coro:main}
    Suppose
    \begin{equation}\label{eq:threshold}
    \begin{aligned}
    \sigma_{\rm th}(\eta)
    &\coloneqq
    \sqrt{
    \frac{2(S+1)}{M}
    \log\frac{2(S+1)}{\eta}
    }\\
    &\quad +
    \frac{4}{3M}\log\frac{2(S+1)}{\eta}.
    \end{aligned}
    \end{equation}
    If
    \begin{gather}
        S = O(\chi^2\log(d\chi^2\eta^{-1})\log(\chi^2\eta^{-1})), \\
        M = \Theta\left(d^2 \log (S\eta^{-1}) S^{-1}\lambda_{\rm lb}^{-2}\right),
    \end{gather}
    then with probability at least $1-2\eta$, Algorithm~\ref{alg:rank_certification} succeeds for $\chi_{\rm tar} = \chi$. The total sample complexity is 
    \begin{equation}
        N_{\mathrm{tot}} = \widetilde{O}\left(d^2 \chi^2 \lambda^{-2}_{\rm lb}\right).
    \end{equation}
\end{corollary}

\begin{figure}
\begin{algorithm}[H]
\caption{\label{alg:rank_certification}%
Certifying entanglement dimensionality by projected correlation matrix}
{\small
\hspace{0pt}\textbf{Input:} target state $|\psi\rangle$, target Schmidt number $\chi_{\rm tar}$, number of Pauli strings $S$, shots per setting $M$, failure probability $\eta$.\\
\hspace{-3pt}\textbf{Output:} $1$ if $\sn(|\psi\rangle)$ is certified to be at least $\chi_{\rm tar}$; $0$ if the certification fails.
\begin{algorithmic}[1]

\State Sample $\caA = (a_s)_{s=1}^S$ with each $a_s$ i.i.d. uniformly drawn from
$\mathbb F_2^{2n}\setminus\{0\}$. 

\State For each $a\in \caA $, measure $\tr_{\rmB}(|\psi\>\<\psi|), \tr_{\rmA}(|\psi\>\<\psi|)$ in the eigenbasis of $W_a$ for $M$ shots. Use the data to estimate $(\bfG_{\psi})_{a,0}, (\bfG_{\psi})_{0,a}$.

\State For each $(a, b)\in \caA \times \caA$, measure $|\psi\rangle$ in the eigenbasis of $W_a \otimes W_b$ for $M$ shots. Use the data to estimate $(\bfG_{\psi})_{a,b}$.

\State Construct the projected correlation matrix $\hat{\bfG}_{\psi,\caA}$ according to \eref{eq:projectedCM}.

\State Let $N_\sigma$ be the number of singular values of $\hat{\bfG}_{\psi,\caA}$ exceeding $\sigma_{\mathrm{th}}(\eta/2)$ in \eref{eq:threshold}.
\If{$N_{\sigma} > (\chi_{\rm tar}-1)^2$} \Return $1$. \EndIf

\State Sample a pair of independent random unitaries $U_{\rmA}, U_{\rmB}$ and prepare
$|\psi_U\rangle \coloneqq (U_{\rmA}\otimes U_{\rmB})|\psi\rangle$.

\State Repeat steps 2--5 with $|\psi\rangle$ replaced by $|\psi_U\rangle$, obtaining $N_{\sigma}$ from $\hat{\bfG}_{\psi_U,\caA}$.

\If{$N_{\sigma} > (\chi_{\rm tar}-1)^2$} \Return $1$
\Else{} \Return $0$.
\EndIf

\end{algorithmic}
}
\end{algorithm}
\end{figure}

In SM Sec.~VII B, we further show that Algorithm~\ref{alg:rank_certification} can also be adapted to estimate the actual Schmidt number.

\textit{Robustness to physical noise.}---In the scenario with physical noise,  we still aim to certify the entanglement dimensionality of the original state $|\psi\>$, whereas the actually measured state is obtained from $|\psi\>$ after local noisy channels and randomized preprocessing. Suppose the state after randomized preprocessing is
\begin{equation}
    \rho_{\mathrm{uc}} \coloneqq (U_\rmA\otimes U_\rmB)(\caN_\rmA \otimes \caN_\rmB)(|\psi\>\<\psi|)(U_\rmA^\dag \otimes U_\rmB^\dag),
\end{equation}
where $\caN_\rmA, \caN_\rmB$ are two local Pauli channels. Since Pauli channels act diagonally in the Pauli basis,
the CM transforms by
$\bfG_{\rm uc} = R_\rmA\Lambda_1 \bfG_\psi \Lambda_2 R_\rmB^\top$, where $\Lambda_1,\Lambda_2$ are diagonal matrices, and $R_\rmA, R_\rmB$ are orthogonal matrices associated with $U_\rmA, U_\rmB$. This transformation never increases the rank. Our certificate is therefore sound against arbitrary uncorrelated
Pauli noise, because the event $\rank(\bfG_{\rm uc,\caA}) > (\chi_{\rm tar}-1)^2$ still certifies
$\sn(|\psi\>) \ge \chi_{\rm tar}$. Moreover, the polylogarithmic measurement-setting scaling persists under mild assumptions on the local Pauli noise. A sufficient condition is given in SM Sec.~VII C.

However, the soundness can fail in the presence of globally correlated noise, such as the
\textit{depolarizing noise}. Consider $\rho_{\rm dp} \coloneqq (1-\epsilon)|\psi\>\<\psi| +
\epsilon I/d^2$, whose white-noise component adds a rank-1 term to the projected CM
$\bfG_{{\rm dp},\caA} = (1-\epsilon)\bfG_{\psi,\caA} + \epsilon|0\>\<0|$,
generically inflating the observed rank by one. Hence, we can only conclude $\sn(|\psi\>) \ge m$ when the number of above-threshold singular values equals $m^2+1, m\in\bbN$. 

\textit{Numerical results.}---To demonstrate the effectiveness of our protocol, we design the following numerical experiments. Consider a random \textit{matrix product state} (MPS) with bond dimension $\chi$ on a $2n$-site 1D lattice. An MPS is a quantum state ansatz of the form
\begin{equation}
    \sum_{i_1,\ldots,i_{2n}}\tr[BA_1^{(i_1)}A_2^{(i_2)}\cdots A_{2n}^{(i_{2n})}] |i_1,\ldots,i_{2n}\>,
\end{equation}
where each $A_{s}^{(i_s)}$ is a $\chi\times\chi$ matrix, while $B$ is a rank-1 $\chi\times\chi$ matrix. The random MPS is generated by sampling each $A_{s}^{(i_s)}$ from a complex Gaussian random matrix and then normalizing the full vector. The state has fixed Schmidt number $\chi$ across the central bipartition almost surely. 

In our numerical tests \cite{code2026}, we set $2n = 28, \chi = 4$ and $\eta = 10^{-2}$. The full CM has dimension $2^{28} \times 2^{28}$ and matrix rank $\chi^2=16$. Matrix $B$ has only one non-zero entry $B_{00}=1$. We set $S = 16, 24, 32, 40, 48$. The sampled MPS has a minimal non-zero Schmidt coefficient of approximately $0.026$. Thus we set $\lambda_{\rm lb} = 10^{-2}$. The local Pauli channels are 
\begin{align}
    \caN_\rmA &= \bigotimes_{j=1}^n \caX_j, \quad \caX_j(\rho) \coloneqq (1-\epsilon_\rmA)\rho + \epsilon_\rmA X_j \rho X_j,\\
    \caN_\rmB &= \bigotimes_{j=1}^n \caZ_j,\quad \caZ_j(\rho) \coloneqq (1-\epsilon_\rmB)\rho + \epsilon_\rmB Z_j \rho Z_j,
\end{align}
with $\epsilon_\rmA = \epsilon_\rmB = 0.01$. We use the conservative shot budget corresponding to the smallest value $S=16$, namely $M = \lfloor \ln(16\times 10^2)\times 2^{28}\times 10^4/16\rfloor $. This is a deliberately conservative theoretical shot budget chosen to test the rigorous threshold. Figure \ref{fig:mps} agrees with the prediction in previous sections: the noise-induced singular values never exceed the threshold, and the randomized preprocessing makes the rank of the projected CM visible and robust.

\begin{figure}
    \centering
    \includegraphics[width=0.9\linewidth]{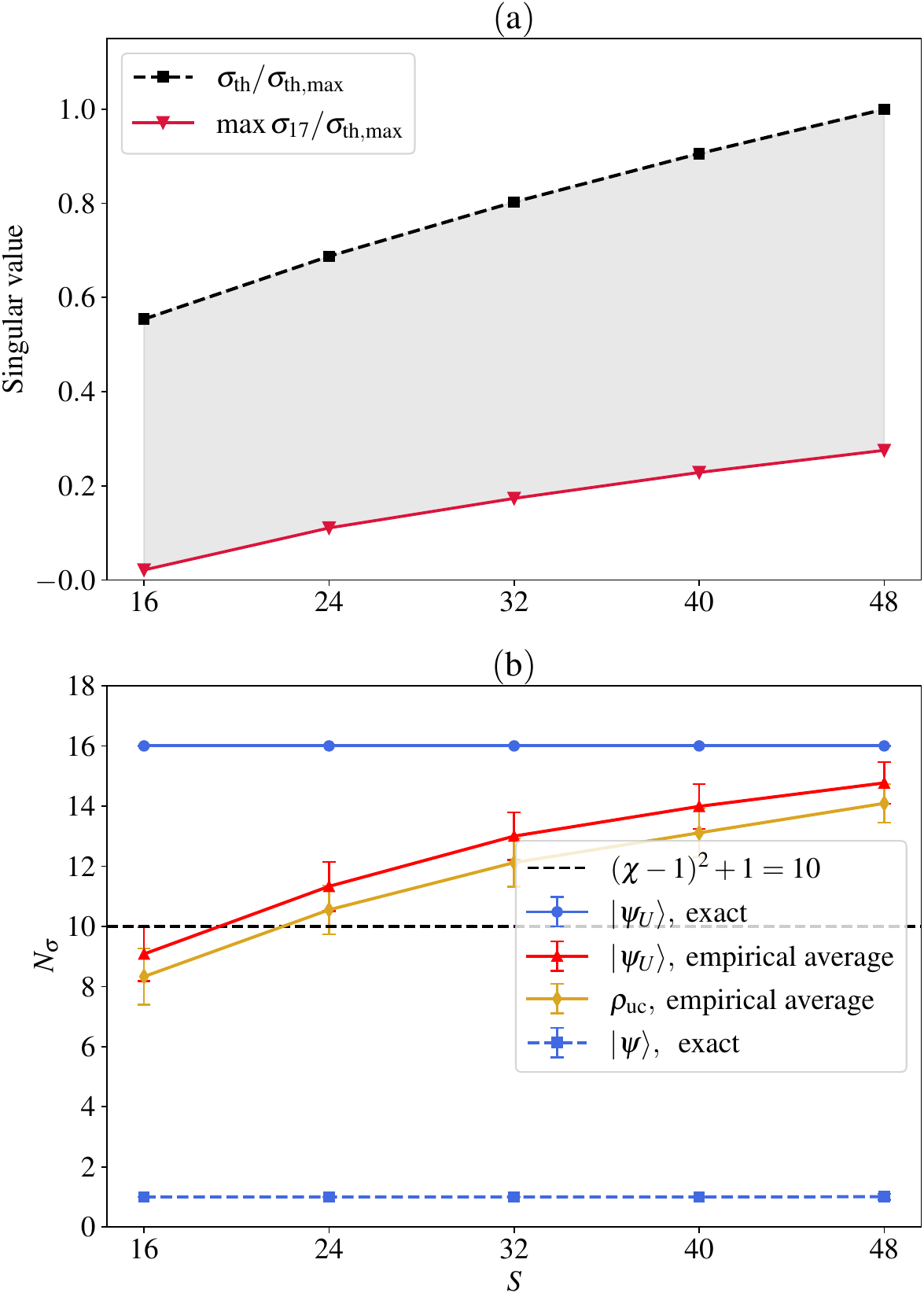}
    \caption{Performance of Algorithm~\ref{alg:rank_certification} on random MPS with $2n = 28, \chi = 4$. Panel (a) compares the maximal $(\chi^2+1)$th singular value of the noisy projected CM among 100 batches and the theoretical threshold value $\sigma_{\rm th}$. The $y$-axis values are normalized by $\sigma_{\rm th,\max}$, the value of $\sigma_{\rm th}$ at $S = 48$. The fact that $\sigma_{\mathrm{th}} > \max\sigma_{17}(\hat{\bfG}_{\psi,\caA})$ means no noise-induced singular value is counted. Panel (b) compares the number of detected singular values of projected CM in different scenarios. The data is averaged over 100 batches. When the number reaches the critical line $N_\sigma = 10$, we can conclude that $\chi \ge 4$. Without the randomized preprocessing, the projected CM almost always has rank 1. With the random unitary layer, the protocol succeeds at $S = 32$ with high probability even in the presence of local Pauli noise.}
    \label{fig:mps}
\end{figure}

\textit{Discussion and conclusion.}---In conclusion, we have developed a basis-independent protocol for certifying the Schmidt number of high-dimensional bipartite pure states using random Pauli measurements and randomized preprocessing. The key observation is that the rank of the correlation matrix is exactly the square of the Schmidt number, so that the entanglement dimensionality can be certified from a much smaller randomly projected correlation matrix. After local scrambling, the number of distinct measurement settings becomes $\tilde O(\chi^4)$, nearly independent of the local Hilbert-space dimension for fixed $\chi$, compared with the $\Omega(d^2)$ worst-case scaling without scrambling. Numerical simulations on random MPS confirm the method’s effectiveness and demonstrate the substantial efficiency gains provided by local random unitaries. 

Our result can be viewed as complementary to fidelity-based certification methods. Fidelity estimation may require only a constant number of measurement settings when the target state and its Schmidt basis are known. However, its certification power can deteriorate under unknown local basis rotations or experimental drift. In contrast, our protocol does not require prior knowledge of the Schmidt basis, remains sound under local Pauli noise, and requires only $\poly(n)$ measurement settings for bounded Schmidt number. It is therefore particularly useful when the target state is difficult to predict or represent classically.

Beyond Schmidt-number certification, the same projected CM viewpoint may also be useful for testing other structural and dynamical properties, including \textit{nonstabilizerness} \cite{oliviero2022measuring} and operator spreading \cite{adam2018operator}. Notably, the incoherence parameter $\mu(U)$ bears a conceptual resemblance to the \textit{stabilizer Rényi entropy} \cite{leone2022stabilizer,lami2023nonstabilizerness}, a computable measure of nonstabilizerness. We briefly discuss its connection to \textit{out-of-time-ordered correlators} in SM Sec.~VIII. Potential deeper connections between these topics deserve further investigation.

\textit{Acknowledgments.}---We thank Huangjun Zhu for insightful discussions. This work is supported by the Shanghai Municipal Science and Technology Major Project (Grant No. 2019SHZDZX01-ZX04) and the institutional research funding of Shanghai University under Project No. N.14-G204-26-603.

The author used GPT-5.6 Sol to assist with language editing and with checking and refining mathematical arguments. All AI-assisted outputs were critically reviewed and revised by the author, who takes full responsibility for the scientific content and integrity of the manuscript.

\bibliography{main}

\clearpage

\onecolumngrid

\setcounter{section}{0}
\setcounter{subsection}{0}
\setcounter{equation}{0}
\setcounter{figure}{0}
\setcounter{table}{0}

\renewcommand{\thesection}{S\arabic{section}}
\renewcommand{\thesubsection}{S\arabic{section}.\arabic{subsection}}
\renewcommand{\theequation}{S\arabic{equation}}
\renewcommand{\thefigure}{S\arabic{figure}}
\renewcommand{\thetable}{S\arabic{table}}

\begin{center}
    {\large\bfseries Supplemental Material for\\[0.5em]
    Certifying Unknown-basis Pure-State Entanglement Dimensionality From Few Measurement Settings}
\end{center}

\vspace{1em}

\section{Table of notation}
\renewcommand{\arraystretch}{1.4}
\begin{table}[H]
\centering
\begin{tabular}{c|c}
\hline
Notation & Meaning \\
\hline
$n$ & subsystem qubit number \\
$d$ & subsystem dimension \\
$\chi$ & Schmidt number \\
$\{\lambda_i\}_{i=1}^{\chi}$ & Schmidt coefficients \\
$\{|l_i\>\}_{i=1}^{\chi}, \{|r_i\>\}_{i=1}^{\chi}$ & left/right Schmidt vector basis \\
$\lambda_{\rm lb}$ & Schmidt coefficients lower bound \\
$\caP$ & the full set of Pauli strings \\
$\caP_0$ & $\caP\setminus\{I\}$ \\
$W_a$ & Pauli operator \\
$\bfG_\psi$ & correlation matrix of state $|\psi\>$ \\
$\caA$ & a random sequence of Pauli strings \\
$\caA_0$ & $(0,\caA)$ \\
$\bfG_{\psi,\caA}$ & projected correlation matrix\\
$\mu(U)$ & incoherence parameter \\
$S$ & the size of the random set $\caA$ \\
$M$ & shots for each measurement basis\\
$\eta$ & failure probability \\
\hline
\end{tabular}
\end{table}

\newpage

\section{Preliminaries}

\subsection{Pauli strings}

The nontrivial single-qubit Pauli operators contain three elements:
\begin{equation}
    X = \left(\begin{array}{cc}
        0 & 1 \\
        1 & 0
    \end{array}\right),\quad Y = \left(\begin{array}{cc}
        0 & -\rmi \\
        \rmi & 0
    \end{array}\right),\quad Z = \left(\begin{array}{cc}
        1 & 0 \\
        0 & -1
    \end{array}\right).
\end{equation}
Suppose $\caH$ is the Hilbert space of $n$ qubits. Then the $4^n$ Hermitian Pauli strings form an operator basis for all Hermitian operators on $\caH$, which can be represented in the Weyl form
\begin{equation}\label{eq:weyl}
    W_{p,q} \coloneqq \bigotimes_{j=1}^{n}\rmi^{p_j q_j}X_j^{p_j}Z_j^{q_j}\qquad p,q\in \bbF_2^n,
\end{equation}
where $X_j, Z_j$ are the Pauli $X, Z$ operators on the $j$th site. 

For any two binary strings $p,q\in \bbF_2^n$, we define
\begin{equation}
    p\cdot q \coloneqq \sum_{j=1}^n p_j q_j,\quad (p+q)_j \coloneqq p_j + q_j \ (\mathrm{mod} \ 2)\quad \forall j = 1,\ldots,n.
\end{equation}
Let $|x\>, x\in \bbF_2^n$ be an element of the computational basis. Then
\begin{equation}
    W_{p,q}|x\> = \rmi^{p\cdot q}(-1)^{q\cdot x}|p+x\>.
\end{equation}

\begin{proposition} \label{prop:swap}
    The operator 
        $\sum_{p,q\in\bbF_2^n} W_{p,q} \otimes W_{p,q}/d$
    is the swap operator on $\caH \otimes \caH$.
\end{proposition}

\begin{proof}
    Let $|x\>,|y\>$ be two arbitrary computational basis vectors. Then 
    \begin{align}
        \left(\frac{1}{d}\sum_{p,q\in\bbF_2^n} W_{p,q} \otimes W_{p,q}\right) |x\> \otimes |y\> &= \frac{1}{d}\sum_{p,q\in\bbF_2^n} W_{p,q}|x\> \otimes W_{p,q}|y\>\nonumber \\
        &= \frac{1}{d}\sum_{p,q\in\bbF_2^n} \rmi^{p\cdot q}(-1)^{q\cdot x}\rmi^{p\cdot q}(-1)^{q\cdot y}|p+x\> \otimes |p+y\> \nonumber \\
        &= \frac{1}{d}\sum_{p,q\in\bbF_2^n} (-1)^{q\cdot(x+y+p)} |p+x\> \otimes |p+y\> \nonumber \\
        &= \sum_{p\in \bbF_2^n} \delta_{x+y+p, 0} |p+x\> \otimes |p+y\> \nonumber \\
        &= |x+y+x\> \otimes |x+y+y\> = |y\>\otimes |x\>.
    \end{align}
\end{proof}

Two Pauli strings $W_a,W_b$ either commute or anticommute with each other. This relation is determined by the symplectic inner product
\begin{equation}
    \gamma(a,b) \coloneqq \sum_{j=1}^n (a_{j}b_{j+n} + a_{j+n} b_j)\ (\mathrm{mod} \ 2).
\end{equation}
If $\gamma(a,b) = 0$, then $[W_a,W_b] = 0$; otherwise $\{W_a,W_b\} = 0$. Equivalently, $W_a W_b W_a = (-1)^{\gamma(a,b)}W_b$.

\subsection{Matrix concentration inequality}

In this appendix and the following, we use $\|v\|_2$ to represent the Euclidean $L^2$ norm. Let $M$ be a matrix with dimension $N \times D$, and $N \ge D > 1$.  Let $\{\lambda_i(M^\dag M)\}_{i=1}^D$ be the set of eigenvalues of $M^\dag M$ in nonincreasing order, then the singular values of $M$ are
\begin{equation}
    \left\{\sigma_i(M) \coloneqq \sqrt{\lambda_i(M^\dag M)}\right\}_{i=1}^D.
\end{equation}

Suppose $\caR$ is a finite set. Then we use $R \sim \caR$ to denote that $R$ is a random variable sampled uniformly from $\caR$. If the sampling is repeated $N$ times, we use $R_i$ to denote the $i$th sample. Let $\caV$ be a finite set of vectors in $\bbC^{D}$.

\begin{theorem}\label{thm:singularv}[Theorem 6.1.1 in Ref. \cite{tropp2015introduction}]
    Suppose $\{Z_i\}_{i=1}^S$ is a finite sequence of independent, random matrices with dimension $d_1\times d_2$, and
    \begin{equation}
        \bbE[Z_i] = 0,\qquad \|Z_i\| \le \mu \quad \forall i,\qquad \left\|\sum_{i=1}^S \bbE[Z_i Z_i^\dag]\right\|, \left\|\sum_{i=1}^S \bbE[Z_i^\dag Z_i]\right\| \le \nu.
    \end{equation}
    Then 
    \begin{equation}
        \Pr\left\{\left\|\sum_{i=1}^S Z_i\right\| \ge \varepsilon\right\} \le (d_1 + d_2) \exp\left(- \frac{\varepsilon^2}{2\nu + 2\mu \varepsilon/3}\right).
    \end{equation}
\end{theorem}
Note that the theorem does not require $\{Z_i\}$ to be identically distributed.

If the failure probability is at most $\eta$, then it suffices to let
\begin{equation}\label{eq:eps}
    \varepsilon = \sqrt{2\nu \log\left(\frac{d_1 + d_2}{\eta}\right)} + \frac{2}{3}\mu \log\left(\frac{d_1 + d_2}{\eta}\right).
\end{equation}
Furthermore, if $\nu = \Theta(S\mu^2)$ and $S = \Omega(\log((d_1 + d_2)/\eta))$, then 
\begin{equation}
    \varepsilon = O\left(\mu\sqrt{S\log\left(\frac{d_1 + d_2}{\eta}\right)}\right).
\end{equation}

\subsection{Schmidt number certification by low-rank matrix recovery}

Under the assumption that $|\psi\>$ is a pure state, the Schmidt number $\chi$ of $|\psi\>$ equals the rank of $\rho_A$. Therefore, the Schmidt number certification can be accomplished by full state tomography of subsystem A as well. In low-rank matrix recovery, we first choose a set of observables $\{E_j\}_{j=1}^m$. Define superoperator
\begin{equation}
    \caM(\cdot) \coloneqq \sum_{j=1}^m \tr(\cdot E_j) |e_j\>.
\end{equation}
In the experiments, we use $O(\epsilon^{-2})$ shots to estimate each $\tr(\rho_A E_j)$ up to accuracy $O(\epsilon)$, and denote the estimator as $\hat{a}_j$. Let $\caL^{\rmH}(\caH_A)$ denote the set of Hermitian operators on $\caH_A$. The matrix recovery is accomplished by solving a convex relaxation task:
\begin{equation}
    \min_{\Theta\in \caL^{\rmH}(\caH_A)} \|\Theta\|_1\quad \mathrm{s.t.}\quad \|\caM(\Theta) - \sum_j \hat{a}_j |e_j\>\|_2 \le \sqrt{m}\epsilon
\end{equation}

\begin{theorem}[Random Pauli measurement \cite{gross2010quantum,liu2011universal}]
    If $m = \widetilde{O}(\chi d)$, then with constant probability, we have
\begin{equation}
    \|\Theta_\ast - \rho_A\|_F = O(\epsilon \sqrt{d}). 
\end{equation}
\end{theorem}

Without further assumption, the spectral norm $\|\Theta_\ast - \rho_A\|_{\infty}$ has upper bound $O(\epsilon \sqrt{d})$, and the perturbation of spectrum is bounded by $O(\epsilon\sqrt{d})$ as well. Hence, we need $\epsilon \sqrt{d} \lesssim \lambda_\lb$
to ensure faithful recovery of the rank. Totally, there are $\widetilde{O}(d\chi)$ different measurement settings, and $O(\epsilon^{-2}) = O(d \lambda_{\rm lb}^{-2})$ shots per setting.

\subsection{Schmidt number certification by fidelity-based criterion}

A well-known Schmidt number certification criterion is based on the fact that if $\sn(\psi) \le \chi$, then
\begin{equation}
    |\<\psi|+_d\>|^2 \le \frac{\chi}{d},
\end{equation}
where $|+_d\>$ is a maximally entangled state on $\caH_A \otimes \caH_B$. The upper bound is saturated iff
\begin{equation}
    |\psi\> = \sum_{i=1}^{\chi}\frac{e^{\rmi\theta_i}}{\sqrt{\chi}}|l_i\>\otimes|r_i\>,
\end{equation}
and $|+_d\>$ has decomposition
\begin{equation}
    |+_d\> = \sum_{i=0}^{d-1}\frac{e^{\rmi\phi_i}}{\sqrt{d}}|l_i\>\otimes |r_i\>.
\end{equation}
Even when $\sn(\psi)=\chi$, the fidelity can be far below $\chi/d$ if the Schmidt bases are misaligned. Hence, we need to approximately know the Schmidt basis vectors of $|\psi\>$. This is the reason why it is called a basis-dependent criterion. The advantage is a single POVM 
\begin{equation}\label{eq:simplePOVM}
    \{|+_d\>\<+_d|, I - |+_d\>\<+_d|\}
\end{equation}
is sufficient to estimate this fidelity. This POVM is hard to realize in experiments. Various protocols have been proposed  \cite{bavaresco2018measurements, morelli2023resource} to estimate the fidelity $F(\psi) \coloneqq |\<\psi|+_d\>|^2$ with only a constant number of measurement settings. However, because the quantity requires accuracy $O(d^{-1})$ to distinguish $\chi/d$ from $(\chi-1)/d$, the estimation of $F(\psi)$ needs $O(d^2)$ measurements per setting. (For this specific POVM in \eref{eq:simplePOVM}, the measurement budget is only $O(d\chi)$ instead, but the improvement does not hold for general fidelity-based methods.)

\subsection{Approximate unitary designs}

For a probability distribution $\mu$ on the
unitary group $\bbU(d)$ we write
\begin{equation}
  \Phi^{(k)}_{\mu}(X)\;:=\;\bbE_{U\sim\mu}\!\left[U^{\otimes k}X\,(U^{\dagger})^{\otimes k}\right],
  \qquad X\in\caB(\caH^{\otimes k}),
\end{equation}
for the $k$-fold twirling channel, and
\begin{equation}
  \Delta^{(k)}_{\mu}\;:=\;\Phi^{(k)}_{\mu}-\Phi^{(k)}_{\Haar}.
\end{equation}

A distribution $\nu$ on $\bbU(d)$ is an $\varepsilon$-approximate unitary
$k$-design in relative error if
\begin{equation}\label{eq:def-relative}
  (1-\varepsilon)\,\Phi^{(k)}_{\Haar} \le \Phi^{(k)}_{\nu} \le (1+\varepsilon)\,\Phi^{(k)}_{\Haar},
\end{equation}
i.e.\ both $(1+\varepsilon)\Phi^{(k)}_{\Haar}-\Phi^{(k)}_{\nu}$ and
$\Phi^{(k)}_{\nu}-(1-\varepsilon)\Phi^{(k)}_{\Haar}$ are completely positive.

A distribution $\nu$ on $\bbU(d)$ is an $\varepsilon$-approximate unitary
$k$-design in diamond distance if
\begin{equation}\label{eq:def-diamond}
  \left\| \bbE_{U\sim \nu}\left[U^{\otimes k}(\cdot)(U^\dag)^{\otimes k}\right]
  - \bbE_{U\sim\Haar}\left[U^{\otimes k} (\cdot) (U^\dag)^{\otimes k}\right]\right\|_{\diamond}
  \;=\;\bigl\|\Delta^{(k)}_{\nu}\bigr\|_{\diamond}\;\le\; \frac{\varepsilon}{d^{k}} .
\end{equation}
According to Ref. \cite{Brandao2016LocalRandom}, distribution $\nu$ is an $\varepsilon/(2d^k)$-approximate unitary $k$-design in relative error implies it is an $\varepsilon$-approximate unitary $k$-design in diamond distance.

For every integer $1\le k'\le k$
and every $X\in\caB(\caH^{\otimes k'})$, if $\nu$ satisfies \eref{eq:def-diamond}, then
\begin{equation}
  \left\|\bbE_{U\sim \nu}\left[U^{\otimes k'}X(U^\dag)^{\otimes k'}\right]
  - \bbE_{U\sim\Haar}\left[U^{\otimes k'} X (U^\dag)^{\otimes k'}\right]\right\|
  \le d^{\,k'-k}\varepsilon\|X\| .
\end{equation}
We will use this property in the proof of theorems. Next result shows the circuit complexity for $\varepsilon$-approximate unitary $k$-design in diamond distance.

\begin{theorem}\label{thm:kdesign}
An $\varepsilon$-approximate unitary $k$-design on $\bbU(2^{n})$ in normalized
diamond distance can be prepared by a $1$D quantum circuit of depth
\begin{equation}\label{eq:circuit_depth_ub}
  O\left(\left(nk+\log\frac{n}{\varepsilon}\right)k\operatorname{polylog}k\right)
\end{equation}
with no ancilla required. In particular, the depth is
$O(nk^{2}\operatorname{polylog}k)$ whenever $\log(1/\varepsilon)=O(nk)$.
\end{theorem}
\begin{proof}
Corollary~1 of Ref.~\cite{schuster2025random} shows that a
relative-error $\varepsilon'$ approximate unitary $k$-design on $n$ qubits
can be implemented by a $1$D circuit without ancillas and with depth
\begin{equation}
    O\left(
    \log\frac{n}{\varepsilon'}\,
    k\operatorname{polylog}k
    \right).
\end{equation}
Set $\varepsilon'
    =2^{-nk-1}\varepsilon$. By Lemma~3 of Ref.~\cite{Brandao2016LocalRandom}, the distribution $\nu$ satisfies
\begin{equation}
    \|\Delta_\nu^{(k)}\|_\diamond
    \le 2\varepsilon'
    =\frac{\varepsilon}{2^{nk}} = \frac{\varepsilon}{d^k}.
\end{equation}
Moreover, \eref{eq:circuit_depth_ub} is proved using
\begin{equation}
    \log\frac{n}{\varepsilon'}
    =
    O\left(
    nk+\log\frac{n}{\varepsilon}
    \right).
\end{equation}
\end{proof}

\section{Proof of Lemma 1}

\begin{lemma}\label{lem: Tdecomposition}
    Let \(|\psi\> = \sum_{i=1}^{\chi}\sqrt{\lambda_i}|l_i\>\otimes |r_i\>\) be a pure state with Schmidt number \(\chi\), and $\bfG_\psi$ is the corresponding CM. Then \(\rank(\bfG_\psi) = \chi^2\) and its singular values are $\{d\sqrt{\lambda_i \lambda_j}\}_{i,j=1}^{\chi}$.
\end{lemma}

\begin{proof}
By definition, we have
\begin{equation}
    \left(\bfG_\psi\right)_{a,b} = \<\psi|W_a \otimes W_b|\psi\> = \sum_{i,j=1}^{\chi}\sqrt{\lambda_i \lambda_j}\<l_i|W_a|l_j\>\<r_i|W_b|r_j\>,
\end{equation}
from which $\bfG_\psi = dU_L \Lambda_\psi U_R^\dag$ follows immediately, where
\begin{align} \label{eq:lambda}
    \Lambda_\psi &\coloneqq \sum_{i,j=1}^{\chi}\sqrt{\lambda_i \lambda_j}|i,j\>\<i,j|,\qquad |i,j\> \coloneqq |di+j\> = |i\>\otimes|j\>,\\
    U_L &\coloneqq \frac{1}{\sqrt{d}}\sum_{a=0}^{d^2-1}\sum_{i,j=0}^{d-1}\<l_i|W_a|l_j\>|a\>\<i,j|,\\
    U_R &\coloneqq \frac{1}{\sqrt{d}}\sum_{a=0}^{d^2-1}\sum_{i,j=0}^{d-1}\<r_j|W_a|r_i\>|a\>\<i,j|.
\end{align}
If both $U_L,U_R$ are unitary operators, then the singular values of $\bfG_\psi$ are identical to those of $d\Lambda_\psi$, which are $\{d\sqrt{\lambda_i \lambda_j}\}_{i,j=1}^{\chi}$, and rank$(\bfG_\psi) = \chi^2$ accordingly. Next, we will verify this.
    \begin{align}
        U_L^\dag U_L  &= \frac{1}{d}\sum_{a=0}^{d^2-1}\sum_{i_1,j_1=0}^{d-1}\sum_{i_2,j_2=0}^{d-1}\<l_{j_1}|W_a|l_{i_1}\>^*\<l_{j_2}|W_a|l_{i_2}\>|i_1,j_1\>\<i_2,j_2| \nonumber\\
        &= \frac{1}{d}\sum_{a=0}^{d^2-1}\sum_{i_1,j_1=0}^{d-1}\sum_{i_2,j_2=0}^{d-1}\<l_{i_1}|W_a|l_{j_1}\>\<l_{j_2}|W_a|l_{i_2}\>|i_1,j_1\>\<i_2,j_2| \nonumber\\
        &= \sum_{i_1,j_1=0}^{d-1}\sum_{i_2,j_2=0}^{d-1}\<l_{i_1} l_{j_2}|\frac{1}{d}\sum_{a=0}^{d^2-1}W_a \otimes W_a|l_{j_1} l_{i_2}\> |i_1,j_1\>\<i_2,j_2| \nonumber\\
        &= \sum_{i_1,j_1=0}^{d-1}\sum_{i_2,j_2=0}^{d-1}\delta_{i_1,i_2}\delta_{j_1,j_2} |i_1,j_1\>\<i_2,j_2| \nonumber\\
        &= \sum_{i,j=0}^{d-1}|i,j\>\<i,j| = I.\\
      U_L U_L^\dag  &= \frac{1}{d}\sum_{i,j=0}^{d-1}\sum_{a,b=0}^{d^2-1}\<l_i|W_a|l_j\>^*\<l_i|W_b|l_j\> |a\>\<b| \nonumber\\
        &= \frac{1}{d}\sum_{i,j=0}^{d-1}\sum_{a,b=0}^{d^2-1}\<l_j|W_a|l_i\>\<l_i|W_b|l_j\> |a\>\<b| \nonumber\\
        &= \frac{1}{d}\sum_{a,b=0}^{d^2-1}\tr(W_a W_b)|a\>\<b| \nonumber\\
        &= \sum_{a,b=0}^{d^2-1}\delta_{a,b}|a\>\<b| = \sum_{a=0}^{d^2-1}|a\>\<a| = I.
    \end{align}
    In the derivation we have used \pref{prop:swap}. The proof for $U_R$ is identical. This completes the proof of \lref{lem: Tdecomposition}.
\end{proof}

\section{Proof of Theorem 1}
\label{app:main_theorem}

\subsection{Outline of proof}

Recall that $\caA \coloneqq (a_1, \ldots, a_S)$, where each $a_s$ is sampled i.i.d. uniformly from $\bbF_2^{2n}\setminus\{0\}$, and $a_0 \coloneqq 0$. Note that the projected CM (in Eq. (4)) has decomposition $\bfG_{\psi,\caA} = V_L \Lambda_\psi V_R^\dag$, where $\Lambda_\psi$ was defined in \eref{eq:lambda}, and
\begin{align}
    V_L &\coloneqq \sum_{s=0}^S\sum_{i,j=1}^{\chi}\<l_j|W_{a_s}|l_i\>|s\>\<i,j|,\\
    V_R &\coloneqq \sum_{s=0}^S\sum_{i,j=1}^{\chi}\<r_i|W_{a_s}|r_j\>|s\>\<i,j|.
\end{align}
Since $\lambda_i > 0$ for all $i$, we have rank$(\Lambda_\psi) = \chi^2$. By virtue of \lref{lem:matrix}, when $V_L$ and $V_R$ have full column rank, we have rank$(\bfG_\caA) = \chi^2$ as well. Hence, the question is how large $S$ should be to keep $\mathrm{rank}(V_L) = \mathrm{rank}(V_R) = \chi^2$ with high probability. We start with the study of $\mathrm{rank}(V_L)$.

Introduce
\begin{equation}
    |v_a\> \coloneqq \frac{1}{\sqrt{d}}\sum_{i,j=1}^{\chi}\<l_j|W_{a}|l_i\> |i,j\>\qquad a\in\bbF_2^{2n}.
\end{equation}
Observe that
\begin{gather}
    V_L = \sqrt{d}\sum_{s=1}^S|s\>\<v_{a_s}| + \sqrt{d}|0\>\<v_0|,\\
    \frac{1}{d}V_L^\dag V_L = |v_0\>\<v_0| + \sum_{s=1}^S |v_{a_s}\>\<v_{a_s}|,\\
    \sum_{a=0}^{d^2-1}|v_a\>\<v_a| = \sum_{i,j=1}^{\chi}|i,j\>\<i,j|,
\end{gather}
which means $\{|v_a\>\}_{a=0}^{d^2-1}$ forms a vector frame on a Hilbert space of dimension $\chi^2$. Let $\caA_0 \coloneqq \caA \cup\{0\}$. Hence, the rank of $V_L$ equals the rank of 
    $\{|v_a\> : a\in\caA_0\}$,
which can be written as
\begin{equation}
    \caV_{\caA_0} \coloneqq \{|v_{a_s}\>\}_{s=0}^S.
\end{equation}
The quantitative relation between rank$(\caV_{\caA_0})$ and $S$ can be characterized by the following theorem. The proof is deferred to \aref{app:theoremmain}.
\begin{theorem}\label{theorem:main}
    Let $\caA \coloneqq (a_1, \ldots, a_S)$, where each $a_s$ is sampled i.i.d. uniformly from $\bbF_2^{2n}\setminus\{0\}$, and $S \le d^2-1$. 
     Then a sample complexity of $S = O(\kappa^{-2}\mu(T_l)\log(\chi^2/\eta))$ is sufficient to ensure
    \begin{gather}
         \mathrm{rank}\left(\caV_{\caA_0}\right) = \chi^2,\quad
        \sigma_{\chi^2}\left(V_L\right) \ge \sqrt{(1-\kappa)\frac{dS}{d^2-1}}
    \end{gather}
     with probability at least $1-\eta$.
\end{theorem}
Recall that
 \begin{equation}
     \mu(T_l) = d\max_{a\neq 0}\sum_{i,j=1}^{\chi}|\<j|T_l^\dag W_a T_l|i\>|^2 = d\max_{a\neq 0}\sum_{i,j=1}^{\chi}|\<l_j|W_a|l_i\>|^2.
 \end{equation}

As shown in the following proposition, the parameter is bounded by $\mu(T_l) \le d\chi$. 
\begin{proposition} 
For any Schmidt basis vector $\{|l_i\> = T_l|i\>\}_{i=1}^{\chi}$, we have $\mu(T_l) \le d\chi$.
\end{proposition}
\begin{proof}
    Let $\Pi_l \coloneqq \sum_{i=1}^{\chi}|l_i\>\<l_i|$. Then 
    \begin{equation}
        \mu(T_l) = d\max_{a\neq 0}\Tr[(\Pi_l W_a \Pi_l)^\dag \Pi_l W_a \Pi_l] = d\max_{a\neq 0}\|\Pi_l W_a \Pi_l\|_F^2 \le d\chi \max_{P\in \caP_0}\|\Pi_l W_a \Pi_l\|_2^2 \le d\chi .
    \end{equation}
\end{proof}
The next corollary demonstrates that when we apply local Haar random unitary $U_\rmA$ on subsystems, then $\mu(U_\rmA T_l) = O(\poly\log(d)\chi^2)$ with high probability, so that the protocol achieves success with merely $S = O(\poly\log(d)\chi^2)$ copies without requiring additional modifications. The proof is deferred to  \aref{app:corollarymain}.

\begin{corollary}\label{coro:main}
    Let $\caA \coloneqq (a_1, \ldots, a_S)$, where each $a_s$ is sampled i.i.d. uniformly from $\bbF_2^{2n}\setminus\{0\}$, and $S \le d^2-1$.  Sample a Haar random unitary $U$ and define
    \begin{gather}
        |u_a\> \coloneqq \frac{1}{\sqrt{d}}\sum_{i,j=1}^{\chi} \<j|U^\dag W_a U|i\> |i,j\>\quad a\in\bbF_2^{2n};\qquad \caU_{\caA_0} \coloneqq \{|u_{a_i}\>\}_{i=0}^S \\
        V^{(u)}_L \coloneqq \sum_{s=1}^S\sum_{i,j=1}^{\chi}\<j|U^\dag W_{a_s}U|i\>|s\>\<i,j| + \sum_{i=1}^{\chi}|0\>\<i,i|.
    \end{gather}
    Fix $0 < \kappa < 1, 0 < \eta < 1$. Then a sample complexity of 
    \begin{equation}
        S = O\left(\kappa^{-2}\chi^2 \log\left(\frac{d\chi^2}{\eta}\right)\log\left(\frac{\chi^2}{\eta}\right)\right)
    \end{equation}
    is sufficient to ensure
    \begin{gather}
        \mathrm{rank}\left(\caU_{\caA_0}\right) = \chi^2,\quad
        \sigma_{\min}(V^{(u)}_L) \ge \sqrt{(1-\kappa)\frac{dS}{d^2-1}}
    \end{gather}
    with failure probability at most $\eta$.
\end{corollary}

An auxiliary to be used is the following. The proof is deferred to \aref{app:proof_of_lemma}.

\begin{lemma}\label{lem:matrix}
    Let $B = U AV^\dag$ where $A$ is a $D\times D$ full-rank matrix, and $U,V$ are $N\times D$ matrices ($N > D$) with full column rank. Then $B$ has rank $D$, and
    \begin{equation}\label{eq:C3}
        \sigma_{D}(B) \ge \sigma_{D}(U) \sigma_{D}(A) \sigma_{D}(V).
    \end{equation}
\end{lemma}

Now we can prove Theorem 1 in the main text. Here is the restatement.
\begin{theorem}[Restatement of Theorem 1 in the main text]\label{theorem:main_formal}
    Suppose $|\psi\> = \sum_{j=1}^{\chi}\sqrt{\lambda_j}|l_j\>\otimes |r_j\>$ is a rank-$\chi$ pure state in $\caH_\rmA \otimes \caH_\rmB$ and $d_\rmA = d_\rmB = d$. Assume $\chi^2 \le S < d^2$. Let $\caA = (a_s)_{s=1}^S$ be a sequence of random binary strings with each $a_s$ i.i.d. uniformly drawn from $\bbF^{2n}_2 \setminus \{0\}$. Sample two Haar random unitaries $U_{\rmA}, U_{\rmB}$ on $\caH_\rmA, \caH_\rmB$ separately. Let $\bfG_{\psi_U, \caA}$ be the projected CM of $|\psi_U\> \coloneqq (U_{\rmA} \otimes U_{\rmB})|\psi\>$. Fix $0 < \kappa,\eta < 1$. Then 
\begin{equation}
        S = O\left(\kappa^{-2}\chi^2 \log(d\chi^2\eta^{-1})\log(\chi^2\eta^{-1})\right)
    \end{equation}
    is sufficient to ensure
    \begin{equation}
        \sigma_{\chi^2}(\bfG_{\psi_U,\caA}) \ge (1-\kappa)\frac{dS}{d^2-1}\lambda_{\rm lb},
    \end{equation}
    with failure probability at most $\eta$.
\end{theorem}

\begin{proof}
Apply Corollary~\ref{coro:main} separately to the left and right Schmidt bases with failure probability \(\eta/2\). By a union bound, with probability at least \(1-\eta\), we have
\begin{equation}
\sigma_{\chi^2}(V_L),
\sigma_{\chi^2}(V_R)
\ge
\sqrt{(1-\kappa)\frac{dS}{d^2-1}} .
\end{equation}
Moreover,
\begin{equation}
\sigma_{\chi^2}(\Lambda_\psi)
=
\min\sqrt{\lambda_i\lambda_j}
=
\min\lambda_i \ge \lambda_{\rm lb}.
\end{equation}
Using Lemma~\ref{lem:matrix} for
$\mathbf G_{\psi,\caA}
=
V_L\Lambda_\psi V_R^\dagger$,
we obtain
\begin{equation}
\sigma_{\chi^2}(\mathbf G_{\psi,\caA})
\ge
\sigma_{\chi^2}(V_R)
\sigma_{\chi^2}(\Lambda_\psi)
\sigma_{\chi^2}(V_L)
\ge
(1-\kappa)
\frac{dS}{d^2-1}\lambda_{\rm lb} .
\end{equation}
\end{proof}

\subsection{Proof of \thref{theorem:main}}
\label{app:theoremmain}

\begin{proof}
We first introduce projector
\begin{align}
\Pi_\psi &\coloneqq \sum_{i,j=1}^{\chi}|i,j\>\<i,j|.
\end{align}
Note that vector $|v_0\> = \sum_{i=1}^{\chi}|i,i\>/\sqrt{d}$ is trivial, and
\begin{equation}
    \Pi_\psi = |v_0\>\<v_0| + \sum_{a=1}^{d^2-1}|v_a\>\<v_a|.
\end{equation}
Let $|w_i\> \coloneqq \sqrt{d^2-1}|v_{a_i}\>$,
    which is a random vector drawn from
    \begin{equation}
        \caW \coloneqq \left\{\sqrt{d^2-1}|v_a\> : a\in \bbF_2^{2n}\setminus\{0\}\right\}.
    \end{equation}
        Its expectation value is
    \begin{gather}
        \bbE[|w_i\>\<w_i|] = \sum_{a=1}^{d^2-1}|v_a\>\<v_a| = \Pi_\psi - |v_0\>\<v_0|.
    \end{gather}
    The left component $V_L$ then satisfies
    \begin{equation}
        \frac{1}{d}V_L^\dag V_L = |v_0\>\<v_0| + \frac{1}{d^2-1}\sum_{i=1}^S |w_i\>\<w_i|. 
    \end{equation}
    Let
    \begin{equation}
        Z_i \coloneqq |w_i\>\<w_i| - \bbE[|w_i\>\<w_i|],
    \end{equation}
    which satisfies $\bbE[Z_i] = 0$. Then
    \begin{align}
        V_L^\dag V_L &= d|v_0\>\<v_0| + \frac{d}{d^2-1}\sum_{i=1}^S Z_i + \frac{dS(\Pi_\psi - |v_0\>\<v_0|)}{d^2-1}, \\
        \bbE[V_L^\dag V_L] &= \frac{dS}{d^2-1}\Pi_\psi + d \frac{d^2-S-1}{d^2-1}|v_0\>\<v_0|,\\
        \operatorname{spec}\left( \bbE[V_L^\dag V_L] \right) &= \left\{\frac{dS}{d^2-1}\right\}^{\times (\chi^2-1)} \sqcup \left\{\frac{dS}{d^2-1} + \chi \frac{d^2-S-1}{d^2-1}\right\}. 
    \end{align}
    Since $S \le d^2-1$, we have
    \begin{equation}
        \sigma_{\chi^2}(\bbE[V_L^\dag V_L]) = \frac{dS}{d^2-1}.
    \end{equation}
    Then we bound $d\sum_{i=1}^S Z_i/(d^2-1)$. In order to apply the matrix Bernstein inequality, we need the upper bounds on $\|Z_i\|$ and $\|\sum_i\bbE[Z_i^2]\|$. For any $a\in\bbF_2^{2n}\setminus\{0\}$, note that
    \begin{equation}
        (d^2-1)\<v_a|v_a\> = \frac{d^2-1}{d}\sum_{i,j=1}^{\chi}|\<l_i|W_a|l_j\>|^2 \le \mu(T_l).
    \end{equation}
    Hence, $\|w_i\|_2 \le \sqrt{\mu(T_l)}$ almost surely, which leads to
    \begin{equation}
        \|Z_i\| \le \||w_i\>\<w_i|\| + \|\bbE[|w_i\>\<w_i|]\| \le \mu(T_l) + 1,
    \end{equation}
    and
    \begin{equation}
        \begin{aligned}
            \bbE[Z_i^2] &= \bbE[|w_i\>\<w_i|w_i\>\<w_i|] - \bbE[|w_i\>\<w_i|]^2 \\
            &\le \mu(T_l)\bbE[|w_i\>\<w_i|] - \bbE[|w_i\>\<w_i|]^2 \\
            &= \mu(T_l)(\Pi_\psi - |v_0\>\<v_0|) - (\Pi_\psi - |v_0\>\<v_0|)^2.
        \end{aligned}
    \end{equation}
    Hence, $\|\bbE[Z_i^2]\| < \mu(T_l)$, and
    \begin{equation}
        \left\|\sum_{i=1}^S \bbE[Z_i^2]\right\| < S\mu(T_l).
    \end{equation}

    According to matrix Bernstein inequality in \thref{thm:singularv} ($D \to \chi^2, \varepsilon \to \kappa S, \mu \to \mu(T_l) + 1, \nu \to S\mu(T_l)$), we have 
    \begin{equation}
        \Pr\left\{\left\|\sum_{i=1}^S Z_i\right\| \ge \kappa S\right\} \le 2\chi^2 \exp\left[-\frac{\kappa^2 S}{2\mu(T_l) + 2(1+\mu(T_l))\kappa/3}\right].
    \end{equation}
    Therefore, when $S = O(\kappa^{-2}\mu(T_l) \log(\chi^2/\eta))$, with probability at least $1-\eta$, we have 
    \begin{equation}
        \sigma_{\chi^2}(V_L^\dag V_L) \ge \sigma_{\chi^2}(\bbE[V_L^\dag V_L]) - \frac{d}{d^2-1}\left\|\sum_{i=1}^S Z_i\right\| \ge (1-\kappa)\frac{dS}{d^2-1},
    \end{equation}
    which implies rank$(V_L) \ge \chi^2$ and
    \begin{equation}
        \sigma_{\chi^2}\left(V_L\right) \ge \sqrt{(1-\kappa)\frac{dS}{d^2-1}}.
    \end{equation}
    
\end{proof}

\subsection{Proof of \coref{coro:main}}
\label{app:corollarymain}
\begin{lemma} \label{lem:randomweyl}
    Suppose $|v\>$ is a Haar random state on a Hilbert space $\caH$ of dimension $d$, and $0 < \eta < 1/2$. Then
    \begin{equation}
        \Pr\left\{\max_{P\in \caP_0}|\<v|W_a|v\>|^2 \ge \frac{72\pi^3}{d}\ln\left(\frac{d}{\eta}\right)\right\} \le \eta.
    \end{equation}
\end{lemma}

\begin{proof}
    Corollary 22 in Ref. \cite{grewal2023low} states that
    \begin{equation}
        \Pr\left\{\max_{a\neq 0}|\<v|W_a|v\>| \ge \varepsilon\right\} \le 2^{2n+1}\exp\left(-\frac{2^n \varepsilon^2}{36\pi^3}\right).
    \end{equation}
    Now we set $d = 2^n$, and
    \begin{equation}
        \eta = 2d^2 e^{-d \varepsilon^2/c_3},\quad c_3 \coloneqq 36\pi^3.
    \end{equation}
    Therefore, the event is equivalent to
    \begin{equation}
        \max_{a\neq 0}|\<v|W_a|v\>|^2 \ge \varepsilon^2 = \frac{c_3}{d}\ln\left(\frac{2d^2}{\eta}\right).
    \end{equation}
    Finally, as long as $0 < \eta < 1/2$, we have
    \begin{equation}
        36\ln\left(\frac{2d^2}{\eta}\right) \le 72\ln\left(\frac{d}{\eta}\right).
    \end{equation}
\end{proof}

The proof of the original \lref{lem:randomweyl} is based on the Levy's lemma \cite{watrous2018theory}. 

\begin{lemma}[Levy's lemma]\label{lem:levy}
    Suppose $f : \bbS^N\to \bbR$ is an $L$-Lipschitz function, and $v$ is a Haar random vector on a $N$-dimensional Hilbert space. Then
    \begin{equation}
        \Pr\left\{\left|f(v) - \bbE[f]\right| \ge \varepsilon\right\} \le 2\exp\left(-\frac{N\varepsilon^2}{9 L^2}\right).
    \end{equation}
\end{lemma}

Then we have the following lemma as a generalization of \lref{lem:randomweyl}, which can be proved by constructing proper Lipschitz functions and applying \lref{lem:levy}.

\begin{lemma}\label{lem:muU}
Let \(U\) be Haar random on a \(d\)-dimensional Hilbert space, and \(\chi\ll d\). Define
\begin{equation}
\mu(U)
\coloneqq
d\max_{a\neq 0}
\sum_{j,k=1}^{\chi}
|\langle j|U^\dagger W_a U|k\rangle|^2 .
\end{equation}
Then with probability at least \(1-\eta\),
\begin{equation}
\mu(U)
=
O\left(
\chi^2\log\left(\frac{d\chi^2}{\eta}\right)
\right).
\end{equation}
\end{lemma}

\begin{proof}
By definition, \(\{U|j\>\}_{j=1}^{\chi}\) are the first \(\chi\) vectors of a Haar-random unitary. By virtue of \lref{lem:randomweyl}, for each $j$, we have
\begin{equation}\label{eq:bound1}
\Pr\left\{\max_{a\neq 0}
|\<j|U^\dag W_a U|j\>|^2
\ge \frac{72\pi^3}{d}\ln\left(\frac{d}{\eta}\right)\right\} \le \eta.
\end{equation}
For each \(j\neq k\), we introduce
\begin{equation}
|u_{jk}^{(n)}\rangle
\coloneqq
\frac{U|j\>+\rmi^n  U|k\>}{\sqrt2} = U \frac{|j\rangle+\rmi^n|k\rangle}{\sqrt2}.
\end{equation}
Note that
\begin{equation}
|\<j|U^\dag P  U|k\>|
\le
\frac12
\sum_{n=0}^{3}
|\langle u_{jk}^{(n)}|W_a|u_{jk}^{(n)}\rangle|.
\end{equation}
For each fixed pair \((j,k)\) and \(n\), the vector
\(|u_{jk}^{(n)}\rangle\) is Haar distributed as well. Apply Lemma~\ref{lem:randomweyl} to each $|u_{jk}^{(n)}\>$. Taking a union bound over all 4 values of \(n\), we obtain
\begin{equation}
\Pr\left\{\max_{P\in\mathcal P_0} \max_n \left|\<u_{jk}^{(n)}|W_a|u_{jk}^{(n)}\>\right|^2
\ge \frac{72\pi^3}{d}\ln\left(\frac{d}{\eta}\right)\right\} \le 4\eta.
\end{equation}
The bound leads to
\begin{equation}
\Pr\left\{\max_{P\in\mathcal P_0} |\<u_j|W_a|u_k\>|^2
\ge \frac{288\pi^3}{d}\ln\left(\frac{d}{\eta}\right)\right\} \le 4\eta.
\end{equation}
In conjunction with \eref{eq:bound1}, we have
\begin{equation}
\Pr\left\{\max_{P\in\mathcal P_0}\max_{j,k}
|\<j|U^\dag P  U|k\>|^2
\ge \frac{288\pi^3}{d}\ln\left(\frac{d}{\eta}\right)\right\} \le 4\chi^2\eta.
\end{equation}
Equivalently, with probability at least $1-\eta'$, 
\begin{equation}
\max_{a\neq 0}\max_{j,k}
|\<j|U^\dag W_a  U|k\>|^2
< \frac{288\pi^3}{d}\ln\left(\frac{4d\chi^2}{\eta'}\right).
\end{equation}
Consequently, with the same success rate, 
\begin{align}
\mu(U)
=
d\max_{a\neq 0}
\sum_{j,k=1}^{\chi}
|\langle j|U^\dag W_aU|k\>|^2
\le
d\chi^2
\max_{a\neq 0}
\max_{j,k}
|\<j|U^\dag W_a  U|k\>|^2
=
O\left(\chi^2
\log\left(\frac{d\chi^2}{\eta'}\right)\right).
\end{align}
The proof is completed by replacing $\eta'$ with $\eta$.
\end{proof}

Corollary~\ref{coro:main} follows by applying
Theorem~\ref{theorem:main} on the event of
Lemma~\ref{lem:muU}. Taking a union bound over the
failure probabilities and absorbing constant factors into
\(\eta\) yields the stated result.

\subsection{Proof of \lref{lem:matrix}}
\label{app:proof_of_lemma}

\begin{proof}
    Clearly rank$(B) \le D$. All we need is to prove rank$(B) \ge D$. According to Sylvester's rank inequality, we have
    \begin{equation}
        \mathrm{rank}(U A) + \mathrm{rank}(V) - D \le \mathrm{rank}(B).
    \end{equation}
    Because $A$ is invertible, we further have rank$(U A) = D$. Therefore, rank$(B) \ge D$, thus rank$(B) = D$.
    
    We continue with the proof of \eref{eq:C3}. Let $x_0$ be a unit vector corresponding to the smallest non-zero singular value of $B$, which is $\sigma_D(B)$. Then we have
    \begin{equation}
        \sigma_{D}(B) = \|Bx_0\|_2 = \|U AV^\dag x_0\|_2 \ge \sigma_{D}(U)\sigma_{D}(A)\|V^\dag x_0\|_2.
    \end{equation}
    Because $x_0\in\mathrm{range}(B^\dag) \subset \mathrm{range}(V)$, it holds that
    \begin{equation}
        \|V^\dag x_0\|_2 \ge \sigma_D(V)\|x_0\|_2.
    \end{equation}
    This completes the proof of \lref{lem:matrix}.
\end{proof}

\section{Proof of Theorem 2}
\label{app:computational}

The statement of Theorem 2 in the main text is as follows.

\begin{theorem}[Restatement of Theorem 2 in the main text]
    Suppose $\chi\ge 2$. Let $\{|l_i\>\}_{i=1}^{\chi}, \{|r_i\>\}_{i=1}^{\chi}$ be elements of the computational basis (Pauli-$Z$ basis). In order to have 
    $\mathrm{rank}\left(\bfG_{\psi,\caA}\right) = \chi^2$
with constant probability, we need $S = \Omega(d)$.
\end{theorem}

Recall that
\begin{equation}
    |v_a\> \coloneqq \frac{1}{\sqrt{d}}\sum_{i,j=1}^{\chi}\<l_j|W_a|l_i\>|i,j\>,
\end{equation}
and $\caA$ is the sequence of Pauli strings. We describe the sampling as follows: we first sample $p\sim \bbF_2^n$, then sample $q\sim \bbF_2^n$. If $p=q=0$, then discard the sample. Hence, it is equivalent to prove the following theorem instead.
\begin{theorem}
    Let $\{|l_i\>\}_{i=1}^{\chi}, \chi\ge 2$ be elements of the computational basis. In order to have \begin{equation}
    \mathrm{rank}\left(\{|v_a\> : a \in \caA\}\right) = \chi^2
\end{equation}
with constant probability, we need $S = |\caA| = \Omega(d)$.
\end{theorem}

\begin{proof}

 Let $W_{p,q}$ be the Weyl representation of the Pauli operators (\eref{eq:weyl}), and write each $a \in \caA$ as $(p,q)$. 
 
 Define $\caL \coloneqq \{l_i\}_{i=1}^{\chi} \subset \bbF_2^n$. Let
\begin{equation}
    \caL_{p} \coloneqq \{l + p : l \in \caL\},\qquad \caL + \caL \coloneqq \{ m + k : m,k\in\caL\}.  
\end{equation}
Here $\caL + \caL$ is formally called the sumset of $\caL$. For each $m = l_i, k = l_j$, we define $|f_{m,k}\> \coloneqq |i,j\>$.
Let $a = (p,q)$. Then we can rewrite $|v_a\>$ as
\begin{align}
    |v_{(p,q)}\> &= \sum_{m,k\in\caL}\<m|W_{p,q}|k\>|f_{m,k}\> = \rmi^{p\cdot q}\sum_{m,k\in\caL}(-1)^{q\cdot k}\delta_{m,p+k}|f_{m,k}\>\nonumber\\
    &= \delta_{p\in\caL + \caL}(-\rmi)^{p\cdot q}\sum_{m\in \caL \cap \caL_p}(-1)^{q\cdot m}|f_{m,m+p}\>.
\end{align}
Note that $|v_{(p,q)}\>$ is non-zero only if $p\in \caL + \caL$. For each $p$, define
 $\caQ_p \coloneqq \{q : (p,q)\in \caA\}$.
Observe that $\<v_{(p_1,q_1)}|v_{(p_2,q_2)}\> \neq 0$ only if $p_1 = p_2$. Thus, we can classify the elements of $\caA$ by
\begin{equation}
    \caA = \caA_+ \sqcup \caA_-,\qquad \caA_+ \coloneqq \bigsqcup_{p\in \caL + \caL}\{(p,q) : q \in \caQ_p\}.
\end{equation}
The elements in $\caA_-$ do not contribute to the rank.

Note that
\begin{equation}
    \mathrm{rank}\left(\{|v_a\> : a\in \caA_+\}\right) = \sum_{p\in \caL + \caL} \mathrm{rank}\left(\{|v_{(p,q)}\> : q \in \caQ_p\}\right).
\end{equation}
Fix $p \in \caL + \caL$. For any $m_1, m_2 \in \caL \cap \caL_p$, we have
\begin{equation}
    \sum_{q=0}^{d-1}(-1)^{q\cdot m_1} (-1)^{q\cdot m_2} = d \delta_{m_1, m_2}.
\end{equation}
Therefore, the vector set $\{|v_{(p,q)}\> : q \in \caQ_p\}$ satisfies
    \begin{equation}
        \sum_{q=0}^{d-1}|v_{(p,q)}\>\<v_{(p,q)}| = d \sum_{m\in \caL \cap \caL_p}|f_{m,m+p}\>\<f_{m,m+p}|,
    \end{equation}
which means the maximal rank of $\{|v_{(p,q)}\> : q \in \caQ_p\}$ is $|\caL \cap \caL_p|$. In order to attain maximal rank, we need at least $|\caL \cap \caL_p|$ elements in $\caQ_p$. Note that
\begin{equation}
    \chi^2 = |\caL|^2 = \sum_{p\in \caL +\caL}|\caL \cap \caL_p|.
\end{equation}
Accordingly, the event $\mathrm{rank}\left(\{|v_a\> : a\in \caA_+\}\right)  = \chi^2$ implies
\begin{equation}
    |\caQ_p| \ge |\caL \cap \caL_p|\qquad \forall p\in (\caL + \caL),
\end{equation}
which further implies
\begin{equation} \label{eq:event}
    \sum_{p\in (\caL+\caL)\setminus\{0\}} |\caQ_p| \ge \sum_{p\in (\caL+\caL)\setminus\{0\}} |\caL \cap \caL_p| = \chi^2 - \chi.
\end{equation}
The LHS equals the number of Pauli strings with $p\in \caL + \caL \setminus\{0\}$. Therefore,
\begin{equation}
    \Pr\left\{\mathrm{rank}\left(\{|v_a\> : a\in \caA_+\}\right)  = \chi^2\right\} \le \Pr\left\{|\{(p,q) \in \caA_+,\ p\in \caL + \caL \setminus\{0\}\}| \ge \chi^2 - \chi\right\}
\end{equation}
For each sample $a = (p,q)$, the probability of having $p\in (\caL+\caL)\setminus\{0\}$ is 
\begin{equation}
    \frac{d(|\caL+\caL|-1)}{d^2-1} = \Theta\left(\frac{|\caL + \caL|}{d}\right).
\end{equation}
Therefore, by virtue of Markov's inequality, in order to satisfy \eref{eq:event}, we need at least
\begin{equation}
    S = \Omega\left(\frac{\chi^2-\chi}{2}\frac{d}{|\caL+\caL|}\right) = \Omega\left(\frac{d\chi^2}{|\caL + \caL|}\right)
\end{equation}
samples. Since $|\caL + \caL| \le \chi^2$, the sample complexity lower bound is $\Omega(d)$. 
\end{proof}

In the specific case of $|\caL + \caL| = O(\chi)$, the lower bound can be improved to $\Omega(d\chi)$.

\section{Proof of Theorem 3}

We first determine the order of unitary design we need.
\begin{lemma}
\label{lem:haar-moments}
Fix $a\neq 0$ and $i,j\in\{0,\dots,\chi-1\}$. Define
$\alpha_{i,j,a}(U) \coloneqq \<i|U^\dag W_a U|j\>$. Then for every integer $q \ge 3$,
\begin{equation}
  \bbE_{U\sim\mathrm{Haar}}\left[|\alpha_{i,j,a}(U)|^{2q}\right]
  \;\le\;
  \left(\frac{96q}{d}\right)^{q}.
  \label{eq:moment-bound}
\end{equation}
\end{lemma}

\begin{proof}
Since $P$ is a nonidentity Pauli string, by Schur's lemma, we have
$\bbE_{\mathrm{Haar}}[U^\dag P U] = 0$. Thus, 
\begin{equation}
    \bbE_{U\sim\mathrm{Haar}}\left[\alpha_{i,j,a}(U)\right] = 0.
\end{equation}
For $U,W\in\mathbb{U}(d)$,
\begin{align}
  |\alpha_{i,j,a}(U)-\alpha(W)|
  &= \bigl|\<l_i|\,U^\dag P(U-W) + (U-W)^\dag P W\,|l_j\>\bigr| \nonumber \\
  &\le \|P(U-W)\| + \|(U-W)^\dag P\|
  \le 2\|U-W\|
  \le 2\|U-W\|_{F}.
\end{align}
Hence, $\alpha_{i,j,a}(U)$ is $2$-Lipschitz with respect to the Frobenius norm.

By virtue of Levy's lemma on unitary group $\mathbb{U}(d)$, for any real-valued $L$-Lipschitz function $f: \bbU(d) \to \bbR$, we have
\begin{equation}
    \Pr\left\{|f(U) - \bbE_{\mathrm{Haar}}[f(U)]| \ge \varepsilon\right\} \le 2\exp\left(-\frac{d\varepsilon^2}{12L^2}\right).
\end{equation}
Applying this lemma to $f(U) = \Re[\alpha_{i,j,a}(U)], \Im[\alpha_{i,j,a}(U)]$ separately, then
\begin{equation}
    \Pr\left\{\left|\alpha_{i,j,a}(U)\right| \ge \varepsilon\right\} \le 4 \exp\left(-\frac{d \varepsilon^2}{96}\right).
\end{equation}
Finally,
\begin{align}
    \bbE_{U\sim\mathrm{Haar}}\left[\left|\alpha_{i,j,a}(U)\right|^{2q}\right] &= \int_{0}^{\infty} \Pr\left\{|\alpha_{i,j,a}(U)| \ge x\right\} dx^{2q} \nonumber\\
    &\le \int_0^{\infty} 2q x^{2q-1} 4 \exp\left(-\frac{dx^2}{96}\right)dx \nonumber\\
    &= 4 q!\left(\frac{96}{d}\right)^q \le \left(\frac{96 q}{d}\right)^q.
\end{align}
In the last step we used $4q!\le q^{q}$ for all $q \ge 3$.
\end{proof}

Now we can prove the theorem.
\begin{theorem}[Restatement of Theorem 3 in the main text]
\label{thm:design-incoherence}
Let $\eta\in(0,1)$ and let $U$ be drawn from an
$\varepsilon$-approximate unitary $2q$-design on $\mathbb{U}(d)$, and
\begin{equation}
  q = \max\left\{3,\Bigl\lceil \ln\!\frac{2\chi^{2}\,|\caP_0|}{\eta} \Bigr\rceil\right\},
  \qquad
  \varepsilon \le \frac{\eta}{2\chi^{2}|\caP_0|}
  \left(\frac{96e q}{d}\right)^{q}.
\end{equation}
Then
\begin{equation}
  \Pr\left\{ \mu(U) \le 96 e \chi^{2} q \right\}
  \ge 1-\eta,
\end{equation}
i.e.\ $\mu(U) = O(\chi^{2}\log (d\chi\eta^{-1}))$ with probability at least $1-\eta$. The circuit depth for constructing the random unitary is \begin{equation}
    O\left(\log (d) \log^2 \left(\frac{d\chi}{\eta}\right) \mathrm{poly} \log \log \left(\frac{d\chi}{\eta}\right)\right).
\end{equation}
\end{theorem}
\begin{proof}

Fix $a\neq 0$ and $i,j<\chi$ and let
$\beta_{i,j,a}(U) \coloneqq |\alpha_{i,j,a}(U)|^{2}$. Since $a\neq 0$, the definition of  $\varepsilon$-approximate unitary $2q$-design ensures that for all $q'\le q$,
\begin{equation}
    \bbE_{U\sim \caE}[\beta_{i,j,a}(U)^{q'}] \le \bbE_{U\sim \mathrm{Haar}}[\beta_{i,j,a}(U)^{q'}] + \varepsilon.
\end{equation}
In conjunction with Markov's inequality and \lref{lem:haar-moments}, we have
\begin{align}
    \Pr\left\{\beta_{i,j,a}(U) \ge \frac{96eq}{d}\right\} \le \frac{\bbE_{U\sim \caE}[\beta_{i,j,a}(U)^q]}{(96eq/d)^q} \le \frac{\bbE_{U\sim \mathrm{Haar}}[\beta_{i,j,a}(U)^q] + \varepsilon}{(96eq/d)^q} \le e^{-q} + \varepsilon \left(\frac{d}{96eq}\right)^q.
\end{align}
In order to have 
\begin{equation}
    \Pr\left\{\max_{P\in \caP_0}\max_{i,j}\beta_{i,j,a}(U) \ge \frac{96eq}{d}\right\} \le \eta,
\end{equation}
it suffices to let
\begin{equation}
    e^{-q}\le \frac{\eta}{2\chi^2 |\caP_0|},\qquad \varepsilon \left(\frac{d}{96eq}\right)^q \le \frac{\eta}{2\chi^2 |\caP_0|},
\end{equation}
which both hold according to the conditions in the theorem. Consequently, with probability at least $1-\eta$, 
\begin{equation}
  \mu(U) = d\max_{a\neq 0}\sum_{i,j=1}^{\chi}\beta_{i,j,a}(U) \le 96eq\chi^2 = O\left(\chi^2 \log(d\chi\eta^{-1})\right).
\end{equation}
Finally, in conjunction with \thref{thm:kdesign}, the circuit depth for constructing such random unitary is
\begin{equation}
    O\left(\log \left(\frac{\log d}{\varepsilon}\right) q\ \mathrm{poly}\log q\right).
\end{equation}
Since $|\caP_0| = d^2-1, q = \Theta(\log (d\chi\eta^{-1})), \varepsilon = d^{-\Theta(q)}$. After computation, the scaling of circuit depth is
\begin{equation}
    O\left(\log (d) \log^2 \left(\frac{d\chi}{\eta}\right) \mathrm{poly} \log \log \left(\frac{d\chi}{\eta}\right)\right).
\end{equation}

\end{proof}

\section{Robustness analysis}
\label{app:robustness}

\subsection{Proof of Corollary 1}

In our protocol, we estimate each $\bfG_{ab}$ by measuring $|\psi\>$ in the eigenbasis of $W_a \otimes W_b$ for $M$ times (we omit the index $\psi$ for simplicity). Finite-shot noise can be simulated without generating individual measurement
outcomes. Let $x,y$ denote the measurement outcome on subsystem $A,B$ separately. Then the four outcome probabilities are
\begin{equation}
  p_{x,y}
  =
  \<\psi|
    \frac{I + x W_a}{2}\otimes\frac{I + y W_b}{2}
  |\psi\>
  =
  \frac{1}{4}
  \bigl(
    1 + x\,\bfG_{a0} + y\,\bfG_{0b} + xy\,\bfG_{ab}
  \bigr),
  \qquad x,y\in\{\pm1\}.
\end{equation}
Let $M_{x,y}$ denote the number of times obtaining outcome $(x,y)$
among the $M$ shots. The empirical estimators of the entries are
\begin{align}
  \hat{\bfG}_{a0} &= \frac{1}{M}\bigl(M_{1,1} + M_{1,-1} - M_{-1,1} - M_{-1,-1}\bigr), \\
  \hat{\bfG}_{0b} &= \frac{1}{M}\bigl(M_{1,1} - M_{1,-1} + M_{-1,1} - M_{-1,-1}\bigr), \\
  \hat{\bfG}_{ab} &= \frac{1}{M}\bigl(M_{1,1} - M_{1,-1} - M_{-1,1} + M_{-1,-1}\bigr).
\end{align}

We next derive a non-asymptotic threshold for the noise-induced singular values. Let
$\delta\bfG\coloneqq\widehat{\bfG}-\bfG,$
where both matrices have dimension $(S+1)\times(S+1)$, including the
identity row and column. We set
$\widehat{\bfG}_{00}=\bfG_{00}=1$.

\begin{lemma}
\label{lem:shot-noise-threshold}
Suppose
\begin{equation}
\sigma_{\rm th}(\eta)
\coloneqq
\sqrt{
\frac{2(S+1)}{M}
\log\frac{2(S+1)}{\eta}
}
+
\frac{4}{3M}\log\frac{2(S+1)}{\eta}.
\label{eq:threshold-general}
\end{equation}
Then for every
$0<\eta<1$,
\begin{equation}
\Pr\left\{\|\delta\bfG\|
\ge
\sigma_{\rm th}(\eta)
\right\}
\le\eta,
\end{equation}
\end{lemma}

\begin{proof}
Let $\caA_0 \coloneqq (0,a_1,a_2,\ldots,a_S)$. For every $a,b\in\caA_0$ and $m=1,\ldots,M$, let $X_{abm}, Y_{abm}$ denote the measurement outcome on subsystem $A,B$ separately, which satisfy 
\begin{gather}
    \Pr\left\{X_{abm} = x, Y_{abm} = y\right\} = \frac{1}{4}(1 + x\bfG_{a0} + y\bfG_{0b} + xy\bfG_{ab}),\qquad x,y\in\{\pm 1\},\\
    \bbE[X_{abm}] = \bfG_{a0},\quad \bbE[Y_{abm}] = \bfG_{0b},\quad \bbE[X_{abm}Y_{abm}] = \bfG_{ab}.
\end{gather}
Then the noise-induced random
matrix is defined as
\begin{equation}
\begin{aligned}
Z_{abm} \coloneqq \frac{X_{abm}Y_{abm}-\bfG_{ab}}{M} |a\rangle\langle b|.
\end{aligned}
\end{equation}
Note that $X_{0bm} = Y_{a0m} = 1$ with probability 1 for all $a,b,m$. Moreover $\hat{\bfG}_{ab} = \sum_{m=1}^M X_{abm}Y_{abm}/M$, and the random matrices $\{Z_{abm}: a,b\in\caA_0, m = 1,\ldots,M\}$ are independent and satisfy
\begin{equation}
\bbE[Z_{abm}]=0,
\qquad
\delta\bfG = 
\sum_{a,b\in\caA_0}\sum_{m=1}^M Z_{abm}.
\end{equation}
Because all outcomes lie in ${\pm1}$, we have $\|Z_{abm}\| \le 2/M$.
We now bound the matrix variance parameter. 
\begin{equation}
\begin{aligned}
    \sum_{a,b\in\caA_0} \sum_{m=1}^M \bbE[Z_{abm}^\dag Z_{abm}] &= \frac{1}{M}\sum_{a,b\in\caA_0} \left(1 + \bfG_{ab}^2 - 2\bfG_{ab}\bbE[X_{abm}Y_{abm}]\right)|b\>\<b| \\
    &= \frac{1}{M}\sum_{b\in\caA_0}\sum_{a\in\caA_0}(1 - \bfG_{ab}^2)|b\>\<b|.
\end{aligned}
\end{equation}
Hence,
\begin{equation}
    \left\|\sum_{a,b\in\caA_0}\sum_{m=1}^M\bbE[Z_{abm}^\dag Z_{abm}]\right\| \le \frac{S+1}{M}.
\end{equation}
The upper bound for $\sum_{a,b\in\caA_0}\sum_{m=1}^M\bbE[Z_{abm}Z_{abm}^\dag]$ is the same. Hence, by virtue of the matrix Bernstein inequality \thref{thm:singularv} and \eref{eq:eps}, the lemma is proved by choosing
\begin{equation}
    d_1,d_2 \to (S+1),\quad \mu \to \frac{2}{M},\quad  \nu \to \frac{S+1}{M}.
\end{equation}

\end{proof}

Then we can prove Corollary 1 in the main text.

\begin{corollary}[Restatement of Corollary 1 in the main text]
    Suppose
    \begin{equation}\label{eq:threshold}
    \begin{aligned}
    \sigma_{\rm th}(\eta)
    \coloneqq
    \sqrt{
    \frac{2(S+1)}{M}
    \log\frac{2(S+1)}{\eta}
    } +
    \frac{4}{3M}\log\frac{2(S+1)}{\eta}.
    \end{aligned}
    \end{equation}
    If
    \begin{gather}
        S = O(\chi^2\log(d\chi^2\eta^{-1})\log(\chi\eta^{-1})), \qquad
        M = \Theta\left(d^2 \log (S\eta^{-1}) S^{-1}\lambda_{\rm lb}^{-2}\right),
    \end{gather}
    then with probability at least $1-2\eta$, Algorithm 1 succeeds for $\chi_{\rm tar} = \chi$. The total sample complexity is 
    \begin{equation}
        N_{\mathrm{tot}} = \tilde{O}\left(d^2 \chi^2 \lambda^{-2}_{\rm lb}\right).
    \end{equation}
\end{corollary}

\begin{proof}
    In the protocol, we only count the number of singular values of $\hat{\bfG}_{\psi_U,\caA}$ that are larger than $\sigma_{\rm th}(\eta)$. By virtue of \lref{lem:shot-noise-threshold}, we have 
    \begin{equation}
        \|\hat{\bfG}_{\psi_U,\caA} - \bfG_{\psi_U,\caA}\| \le \sigma_{\rm th}(\eta),
    \end{equation}
    with probability at least $1-\eta$. On the other hand, applying \thref{theorem:main_formal} with $\kappa = 1/2$ gives lower bound
    \begin{equation}
        \sigma_{\chi^2}(\bfG_{\psi_U,\caA}) \ge \frac{dS}{2d^2-2}\lambda_{\rm lb}.
    \end{equation}
    Note that $\sigma_{\chi^2+1}(\bfG_{\psi_U,\caA}) = 0$. By virtue of Weyl's inequality, the singular values of $\hat{\bfG}_{\psi_U,\caA}$ satisfy
    \begin{equation}
        \sigma_{\chi^2}(\hat{\bfG}_{\psi_U,\caA}) \ge \frac{dS}{2d^2-2}\lambda_{\rm lb} - \sigma_{\rm th}(\eta),\quad \sigma_{\chi^2+1}(\hat{\bfG}_{\psi_U,\caA}) \le \sigma_{\rm th}(\eta).
    \end{equation}
    In order to detect all $\chi^2$ singular values, we need $\sigma_{\chi^2}(\hat{\bfG}_{\psi_U,\caA}) > \sigma_{\rm th}(\eta)$, which is
    \begin{equation}
        \frac{S\lambda_{\mathrm{lb}}}{d} = \Omega\left(\sqrt{\frac{S\log S}{M}} + \frac{\log S}{M}\right).
    \end{equation}
    It suffices to choose
    \begin{equation}
        M = \Theta\left(\frac{d^2 \log S}{S \lambda_{\rm lb}^2}\right).
    \end{equation}
    Finally, the total sample complexity is given by $N_{\mathrm{tot}} = O(S^2M) = \tilde{O}(d^2\chi^2\lambda_{\rm lb}^{-2})$.
\end{proof}

\textit{Remark.} We actually only need $\sigma_{(\chi-1)^2+1}(\bfG_{\psi_U,\caA}) \ge \sigma_{\rm th}(\eta)$ to conclude that $\chi \ge \chi_{\rm tar}$, but the $[(\chi-1)^2+1]$th singular value is a lot harder to analyze than the smallest singular value.

\subsection{Algorithm for Schmidt number estimation}

Firstly, we observe that Algorithm~1 in the main text can be easily modified into a lower bound estimation algorithm. See Algorithm~\ref{alg:sn_estimation}.

\begin{figure}
\begin{algorithm}[H]
\caption{Schmidt number lower bound estimation by projected correlation matrix}
\label{alg:sn_estimation}
\begin{algorithmic}[1]

\State \textbf{Input:} Bipartite pure state $|\psi\rangle$,
Schmidt number upper bound $\chi_{\rm ub}$, number of Pauli strings $S$,
shots per setting $M$,
failure probability $\eta$.

\State \textbf{Output:} A lower bound of $\sn(|\psi\>)$.

\State Sample $\caA = (a_s)_{s=1}^S$ with each $a_s$ i.i.d. uniformly drawn from
$\mathbb F_2^{2n}\setminus\{0\}$. 

\State Sample independent local random unitaries
$U_{\rm A},U_{\rm B}$ and prepare
$|\psi_U\rangle
=
U_{\rm A}\otimes U_{\rm B}|\psi\rangle$.

\State For each $a\in \caA $, measure $\tr_{\rmB}(|\psi_U\>\<\psi_U|), \tr_{\rmA}(|\psi_U\>\<\psi_U|)$ in the eigenbasis of $W_a$ for $M$ shots. Use the data to estimate $(\bfG_{\psi_U})_{a0}, (\bfG_{\psi_U})_{0a}$.

\State For each $(a, b)\in \caA \times \caA$, measure $|\psi_U\rangle$ in the eigenbasis of $W_a \otimes W_b$ for $M$ shots. Use the data to estimate $(\bfG_{\psi_U})_{ab}$.

\State Construct the projected correlation matrix $\hat{\bfG}_{\psi_U,\caA}$.

\State Let
\[
N_\sigma =
\#\left\{
j:
\sigma_j(\widehat{\mathbf G}_{\psi_U,\caA})
>
\sigma_{\rm th}(\eta)
\right\},\qquad \widehat\chi
=
\max\left\{
1,
\min\left\{
\chi_{\rm ub},
\left\lceil\sqrt{N_\sigma}\right\rceil
\right\}
\right\}.
\]

\State \Return $\hat{\chi}$.

\end{algorithmic}
\end{algorithm}
\end{figure}

Then we can design an adaptive strategy that outputs the most accurate estimate of $\chi$. The full algorithm is illustrated in Algorithm~\ref{alg:adaptive_sn}. Here is how it works. First, by virtue of \thref{theorem:main_formal}, when
\begin{equation}
    S = O\left(\chi_{\rm ub}^2 \operatorname{polylog}\frac{d\chi_{\rm ub}}{\eta}\right),
\end{equation}
we have (let $\kappa=1/2$)
\begin{equation}
    \sigma_{\chi^2}(\bfG) \ge \frac{dS}{2d^2-2}\lambda_{\rm lb},
\end{equation}
for any $\chi \le \chi_{\rm ub}$. Hence, if there exists an integer $t$, such that
\begin{equation}
    \sigma_{t^2+1}(\hat{\bfG}) + \sigma_{\rm th}(\eta) < \frac{dS}{2d^2-2}\lambda_{\rm lb}.
\end{equation}
Then we can conclude that $\sigma_{t^2+1}(\bfG) = 0$,
and $\operatorname{rank}(\bfG_{\psi}) \le t^2$. In this way, we obtain an upper bound of $\chi$. When the lower bound of $\chi$ is identical to the upper bound, we find the true Schmidt number. As we increase $M$, the estimation $\hat{\bfG}$ gets closer to $\bfG$. Eventually, the lower bound and upper bound are both exactly $\chi$.

\begin{figure}
\begin{algorithm}[H]
\caption{Adaptive Schmidt number estimation by projected correlation matrix}
\label{alg:adaptive_sn}
\begin{algorithmic}[1]

\State \textbf{Input:}
Bipartite pure state $|\psi\rangle$, Schmidt number upper bound $\chi_{\rm ub}$, Schmidt coefficient lower bound $\lambda_{\rm lb}$, number of Pauli strings $S$,
initial shot budget $M_0$, failure probability $\eta$.

\State Sample $\caA = (a_s)_{s=1}^S$ with each $a_s$ i.i.d. uniformly drawn from
$\mathbb F_2^{2n}\setminus\{0\}$.

\State Sample independent local random unitaries
$U_{\rm A},U_{\rm B}$ and prepare
$|\psi_U\rangle
=
U_{\rm A}\otimes U_{\rm B}|\psi\rangle$.

\For{$t=0,1,2,\ldots$}

    \State Set $M_t=2^tM_0, \,\eta_t=\eta/[2(t+1)(t+2)]$.

    \State For each $a\in \caA $, measure $\tr_{\rmB}(|\psi_U\>\<\psi_U|),\tr_{\rmA}(|\psi_U\>\<\psi_U|)$ in the eigenbasis of $W_a$ for $M_t$ shots. Use the data to estimate $(\bfG^{(t)}_{\psi_U})_{a0}, (\bfG^{(t)}_{\psi_U})_{0a}$.

\State For each $(a, b)\in \caA \times \caA$, measure $|\psi_U\rangle$ in the eigenbasis of $W_a \otimes W_b$ for $M_t$ shots. Use the data to estimate $(\bfG^{(t)}_{\psi_U})_{ab}$.

\State Construct the projected correlation matrix $\hat{\bfG}^{(t)}_{\psi_U,\caA}$.
    
    \State Let $\sigma_t=\sigma_{\rm th}(\eta_t;S,M_t)$. Let
    \[
    N_{1,t}
    =
    \#
    \left\{
    j:
    \sigma_j(\widehat{\mathbf G}^{(t)}_{\psi_U,\caA})
    \ge \sigma_t
    \right\},\quad N_{2,t} = \inf \left\{j : \sigma_{j+1}(\widehat{\mathbf G}^{(t)}_{\psi_U,\caA})
    +\sigma_t
    < \frac{dS}{2(d^2-1)}\lambda_{\rm lb}\right\}.
    \]

    \If{$N_{1,t} = N_{2,t}$}
        \State \Return $\lceil\sqrt{N_{1,t}}\rceil$.
    \EndIf

\EndFor

\end{algorithmic}
\end{algorithm}
\end{figure}

\subsection{Robustness to physical noise}

In this appendix, we discuss the performance of our protocol when the
prepared state is $\caN(|\psi\>\<\psi|)$ instead of the target $|\psi\>$.
We emphasize that the conclusions in the noisy setting concern the pre-noise pure states $|\psi\>$.

\subsubsection{Uncorrelated Pauli noise} 

Consider two arbitrary Pauli noise channels $\caN_\rmA, \caN_\rmB$ acting on $\caH_\rmA, \caH_\rmB$ separately:
\begin{equation}
    \caN_\rmA(\cdot) = \sum_{\alpha}x_\alpha W_\alpha(\cdot) W_\alpha,\quad \caN_\rmB(\cdot) = \sum_{\beta}y_\beta W_\beta (\cdot) W_\beta.
\end{equation}
Let $\caN = \caN_\rmA \otimes \caN_\rmB$ and $\rho_{\rm uc}  \coloneqq (U_\rmA \otimes U_\rmB)\caN(|\psi\>\<\psi|)(U_\rmA^\dag\otimes U_\rmB^\dag)$. Label the CM of $\rho_{\rm uc}$ as $\bfG_{\rm uc}$. Then
\begin{align}
    \bfG_{\mathrm{uc},ab} &= \tr\left[\caN(|\psi\>\<\psi|)(U_\rmA^\dag W_a U_\rmA \otimes U_\rmB^\dag W_b U_\rmB)\right] \nonumber \\
    &= \sum_{i,j=1}^\chi \sqrt{\lambda_i\lambda_j} \tr\left[\caN_\rmA(|l_i\>\<l_j|)U_\rmA^\dag W_a U_\rmA\right] \tr\left[\caN_\rmB(|r_i\>\<r_j|)U_\rmB^\dag W_b U_\rmB)\right] .
\end{align}
We first show that $\operatorname{rank}(\bfG_{\rm uc}) \le \operatorname{rank}(\bfG_\psi)$. Let
\begin{gather}
    U_\rmA^\dag W_a U_\rmA = \sum_{a'}(R_\rmA)_{aa'}W_{a'}\qquad U_\rmB^\dag W_b U_\rmB = \sum_{b'}(R_\rmB)_{bb'}W_{b'},\\
    \Lambda_1 \coloneqq \sum_{a'=0}^{d^2-1}\left[\sum_{\alpha}x_\alpha(-1)^{\gamma(\alpha,a')}\right]|a'\>\<a'|,\qquad \Lambda_2 \coloneqq \sum_{b'=0}^{d^2-1}\left[\sum_{\beta}y_\beta(-1)^{\gamma(\beta,b')}\right]|b'\>\<b'|.
\end{gather}
Hence, after computation we get $\bfG_{\rm uc} = R_\rmA\Lambda_1\bfG_{\psi}\Lambda_2 R_\rmB^\top$, which gives  $\operatorname{rank}(\bfG_{\rm uc}) \le \operatorname{rank}(\bfG_\psi)$.

Then we analyze the performance of the projected CM protocol in the presence of local Pauli noise. Introduce
\begin{equation}
    |\mu_a\> \coloneqq  \frac{1}{\sqrt{d}}\sum_{i,j=1}^\chi \tr\left[\caN_\rmA(|l_i\>\<l_j|)U_\rmA^\dag W_a U_\rmA)\right] |i,j\>,\quad |\nu_a\> \coloneqq \frac{1}{\sqrt{d}}\sum_{i,j=1}^\chi \tr\left[\caN_\rmB(|r_i\>\<r_j|)U_\rmB^\dag W_a U_\rmB)\right] |j,i\>.
\end{equation}
Then $\bfG_{\rm uc,\caA} = d\Gamma_\rmA \Lambda_\psi \Gamma_\rmB^\dag$, where
\begin{equation}
    \Gamma_\rmA \coloneqq \sum_{s=0}^{S}|s\>\<\mu_{a_s}|,\qquad \Gamma_\rmB \coloneqq \sum_{s=0}^{S}|s\>\<\nu_{a_s}|.
\end{equation}
Hence, the question is reduced to estimating the rank of
\begin{equation}
    \{|\mu_a\> : a\in \caA_0\},\qquad \{|\nu_a\> : a\in \caA_0\}.
\end{equation}

\begin{lemma}\label{lem:noisy_haar_incoherence}
    Suppose $U_\rmA$ is sampled from the Haar random unitary ensemble, then there exists a bound 
    \begin{equation}
        C_\rmA = O\left(\frac{\chi^2}{d^2}\log \frac{d}{\eta}\right),
    \end{equation}
    such that
\begin{equation}
\Pr\left\{\max_{a\neq 0}\<\mu_a|\mu_a\> \le C_\rmA\right\} \ge 1-\eta.
\end{equation}
\end{lemma}

\begin{proof}
    Since $\caN_\rmA(\cdot) = \sum_\alpha x_\alpha W_\alpha (\cdot)W_\alpha$, we have
\begin{align}
\langle\mu_a|\mu_a\rangle
&=
\frac1d
\sum_{i,j=1}^{\chi}
\left|
\langle l_j|
\caN_{\rm A}(U_{\rm A}^\dagger W_aU_{\rm A})
|l_i\rangle
\right|^2
\nonumber\\
&=
\frac1d
\sum_{i,j=1}^{\chi}
\left|
\sum_\alpha x_\alpha
\langle l_j|
W_\alpha U_{\rm A}^\dagger W_aU_{\rm A}W_\alpha
|l_i\rangle
\right|^2.
\end{align}
By Jensen's inequality,
\begin{align}
\langle\mu_a|\mu_a\rangle
&\le
\sum_\alpha x_\alpha
\frac1d
\sum_{i,j=1}^{\chi}
\left|
\langle l_j|
W_\alpha U_{\rm A}^\dagger W_aU_{\rm A}W_\alpha
|l_i\rangle
\right|^2.
\label{eq:mu-jensen}
\end{align}
Define
\begin{equation}
Y_{a,\alpha}
\coloneqq
\frac1d
\sum_{i,j=1}^{\chi}
\left|
\langle l_j|
W_\alpha U_{\rm A}^\dagger W_aU_{\rm A}W_\alpha
|l_i\rangle
\right|^2.
\end{equation}
Then
$\langle\mu_a|\mu_a\rangle
\le
\sum_\alpha x_\alpha Y_{a,\alpha}$.
Since $x\mapsto x^q$ is convex for every $q\ge1$,
\begin{equation}
\langle\mu_a|\mu_a\rangle^q
\le
\sum_\alpha x_\alpha Y_{a,\alpha}^q.
\label{eq:mu-q-jensen}
\end{equation}

For every fixed $\alpha$, define $V_\alpha
\coloneqq
U_{\rm A}W_\alpha$.
Therefore,
\begin{equation}
Y_{a,\alpha}
=
\frac1d
\sum_{i,j=1}^{\chi}
\left|
\langle l_j|
V_\alpha^\dagger W_aV_\alpha
|l_i\rangle
\right|^2.
\end{equation}
Since $U_{\rm A}$ is Haar distributed, $V_\alpha$ is also Haar distributed.

Using
\begin{equation}
\left(
\sum_{k=1}^m z_k
\right)^q
\le
m^{q-1}
\sum_{k=1}^m z_k^q,
\qquad
z_k\ge0,
\end{equation}
with $m=\chi^2$, we obtain
\begin{align}
(Y_{a,\alpha})^q
&\le
\frac{(\chi^2)^{q-1}}{d^q}
\sum_{i,j=1}^{\chi}
\left|
\langle l_j|
V_\alpha^\dagger W_aV_\alpha
|l_i\rangle
\right|^{2q}.
\end{align}
By the Haar moment bound
\begin{equation}
\bbE_U\left[
\left|
\langle l_j|
U^\dagger W_aU
|l_i\rangle
\right|^{2q}
\right]
\le
\left(
\frac{96q}{d}
\right)^q,
\qquad
q\ge3,
\end{equation}
we have
\begin{align}
\bbE_U[Y_{a,\alpha}^q]
&\le
\frac{(\chi^2)^{q-1}}{d^q}
\sum_{i,j=1}^{\chi}
\left(
\frac{96q}{d}
\right)^q =
\left(
\frac{96q\chi^2}{d^2}
\right)^q.
\label{eq:Y-moment}
\end{align}
Combining Eqs.~\eqref{eq:mu-q-jensen} and \eqref{eq:Y-moment}, together with
$\sum_\alpha x_\alpha=1$, yields
\begin{equation}
\bbE_U\left[
\langle\mu_a|\mu_a\rangle^q
\right] \le
\sum_\alpha x_\alpha
\bbE_U[Y_{a,\alpha}^q]
\le
\left(
\frac{96q\chi^2}{d^2}
\right)^q.
\label{eq:mu-moment}
\end{equation}

By Markov's inequality,
\begin{align}
\Pr_U\left\{
\langle\mu_a|\mu_a\rangle
\ge
\frac{96e q\chi^2}{d^2}
\right\}
&\le
\frac{
\bbE_U[
\langle\mu_a|\mu_a\rangle^q
]
}{
(96e q\chi^2/d^2)^q
}\le
e^{-q}.
\end{align}
Follow the method in proving \thref{thm:design-incoherence}. Taking a union bound over all $d^2-1$ nonidentity Pauli operators gives
\begin{equation}
\Pr_U\left\{
\max_{a\ne0}
\langle\mu_a|\mu_a\rangle
\ge
\frac{96e q\chi^2}{d^2}
\right\}
\le
(d^2-1)e^{-q}.
\end{equation}
The lemma is proved by choosing $C_\rmA = 96eq\chi^2d^{-2}$, and 
\begin{equation}
q
=
\max\left\{
3,
\left\lceil
\log\frac{d^2-1}{\eta}
\right\rceil
\right\}.
\end{equation}

\end{proof}

The left component $\Gamma_{\rm A}$ satisfies
\begin{gather}
   \Gamma_{\rm A}
    =
    |0\rangle\langle\mu_0|
    +
    \sum_{s=1}^S |s\rangle\langle\mu_{a_s}|,
    \\
    \Gamma_{\rm A}^\dagger\Gamma_{\rm A}
    =
    |\mu_0\rangle\langle\mu_0|
    +
    \sum_{s=1}^S
    |\mu_{a_s}\rangle\langle\mu_{a_s}|.
\end{gather}

In order to estimate the rank and the minimum singular value of
$\Gamma_{\rm A}$, we bound
$\langle\mu_a|\mu_a\rangle$ and
$\bbE[|\mu_{a_s}\rangle\langle\mu_{a_s}|]$.
Define
\begin{equation}
    \bfP_\psi
    \coloneqq
    \sum_{i_1,j_1,i_2,j_2=1}^{\chi}
    \tr\left[
    \caN_{\rm A}(|l_{i_1}\rangle\langle l_{j_1}|)
    \caN_{\rm A}(|l_{i_2}\rangle\langle l_{j_2}|)^\dagger
    \right]
    |i_1,j_1\rangle\langle i_2,j_2|.
\end{equation}

\begin{theorem} \label{thm:pauli_singular_value}
    Suppose $U_\rmA$ is drawn from the Haar random unitary ensemble, and $\caA = (a_1,a_2,\ldots, a_S)$ is a sequence of random binary strings with each $a_s$ drawn uniformly i.i.d. from $\bbF_2^{2n}\setminus\{0\}$. If 
    \begin{equation}
        S = O\left(\frac{\chi^2\|\bfP_\psi\|}{\kappa^2\lambda^2_{\min}(\bfP_\psi)}\log \frac{d}{\eta}\log\frac{\chi}{\eta}\right)
    \end{equation}
    Then with probability at least $1-\eta$,
    \begin{equation}
        \operatorname{rank}(\Gamma_A) = \chi^2,\qquad \sigma_{\chi^2}(\Gamma_\rmA) \ge \sqrt{(1-\kappa)\frac{S}{d^2-1}\lambda_{\min}(\bfP_\psi)}.
    \end{equation}
\end{theorem}

\begin{proof}

Since
$\{U_{\rm A}^\dagger W_a U_{\rm A}/\sqrt d\}_{a=0}^{d^2-1}$
forms an orthonormal operator basis with respect to the
Hilbert--Schmidt inner product, the Parseval identity gives
\begin{equation}
    \sum_{a=0}^{d^2-1}
    |\mu_a\rangle\langle\mu_a|
    =
    \bfP_\psi.
\end{equation}

For $a=0$, since $\caN_{\rm A}$ is trace preserving,
\begin{equation}
    |\mu_0\rangle
    =
    \frac{1}{\sqrt d}
    \sum_{i=1}^{\chi}|i,i\rangle.
\end{equation}
Therefore,
\begin{equation}
    \bfP_\psi
    =
    |\mu_0\rangle\langle\mu_0|
    +
    \sum_{a=1}^{d^2-1}
    |\mu_a\rangle\langle\mu_a|.
\end{equation}
Since $a_s$ is uniformly sampled from the $d^2-1$ nonidentity
Pauli operators,
\begin{equation}
    \bbE\left[
    |\mu_{a_s}\rangle\langle\mu_{a_s}|
    \right]
    =
    \frac{
    \bfP_\psi
    -
    |\mu_0\rangle\langle\mu_0|
    }{d^2-1}.
\end{equation}
Consequently,
\begin{align}
    \bbE[\Gamma_{\rm A}^\dagger\Gamma_{\rm A}]
    &=
    |\mu_0\rangle\langle\mu_0|
    +
    S\bbE\left[
    |\mu_{a_s}\rangle\langle\mu_{a_s}|
    \right]
    \nonumber\\
    &=
    |\mu_0\rangle\langle\mu_0|
    +
    \frac{S}{d^2-1}
    \left(
    \bfP_\psi
    -
    |\mu_0\rangle\langle\mu_0|
    \right).
\end{align}
Equivalently,
\begin{equation}
    \bbE[\Gamma_{\rm A}^\dagger\Gamma_{\rm A}]
    =
    \frac{S}{d^2-1}\bfP_\psi
    +
    \left(
    1-\frac{S}{d^2-1}
    \right)
    |\mu_0\rangle\langle\mu_0|.
    \label{eq:expected-Gamma-A}
\end{equation}

Let
\begin{equation}
    Z_s
    \coloneqq
    (d^2-1)
    \left(
    |\mu_{a_s}\rangle\langle\mu_{a_s}|
    -
    \bbE\left[
    |\mu_{a_s}\rangle\langle\mu_{a_s}|
    \right]
    \right).
\end{equation}
Then $\bbE[Z_s]=0$, and
\begin{equation}
    |\mu_{a_s}\rangle\langle\mu_{a_s}|
    =
    \bbE\left[
    |\mu_{a_s}\rangle\langle\mu_{a_s}|
    \right]
    +
    \frac{1}{d^2-1}Z_s.
\end{equation}
It follows that
\begin{align}
    \Gamma_{\rm A}^\dagger\Gamma_{\rm A}
    &=
    |\mu_0\rangle\langle\mu_0|
    +
    \sum_{s=1}^S
    |\mu_{a_s}\rangle\langle\mu_{a_s}|
    \nonumber\\
    &=
    |\mu_0\rangle\langle\mu_0|
    +
    S\bbE\left[
    |\mu_{a_s}\rangle\langle\mu_{a_s}|
    \right]
    +
    \frac{1}{d^2-1}
    \sum_{s=1}^S Z_s
    \nonumber\\
    &=
    \bbE[\Gamma_{\rm A}^\dagger\Gamma_{\rm A}]
    +
    \frac{1}{d^2-1}
    \sum_{s=1}^S Z_s.
    \label{eq:Gamma-A-centered}
\end{align}
By Weyl's inequality,
\begin{equation}
    \sigma_{\min}^2(\Gamma_{\rm A})
    =
    \lambda_{\min}
    (\Gamma_{\rm A}^\dagger\Gamma_{\rm A})
    \ge
    \lambda_{\min}
    \left(
    \bbE[\Gamma_{\rm A}^\dagger\Gamma_{\rm A}]
    \right)
    -
    \frac{1}{d^2-1}
    \left\|
    \sum_{s=1}^S Z_s
    \right\|.
    \label{eq:Gamma-A-min-bound}
\end{equation}

In particular, when $S\le d^2-1$, the second term in
Eq.~\eqref{eq:expected-Gamma-A} is positive semidefinite, and hence
\begin{equation}
    \lambda_{\min}
    \left(
    \bbE[\Gamma_{\rm A}^\dagger\Gamma_{\rm A}]
    \right)
    \ge
    \frac{S}{d^2-1}
    \lambda_{\min}(\bfP_\psi).
\end{equation}
Therefore,
\begin{equation}
    \sigma_{\min}^2(\Gamma_{\rm A})
    \ge
    \frac{S}{d^2-1}
    \lambda_{\min}(\bfP_\psi)
    -
    \frac{1}{d^2-1}
    \left\|
    \sum_{s=1}^S Z_s
    \right\|.
    \label{eq:Gamma-A-final-bound}
\end{equation}

 In order to apply the matrix Bernstein inequality, we need the upper bounds on $\|Z_s\|$ and $\|\sum_s\bbE[Z_s^2]\|$. For any $a\in\bbF_2^{2n}\setminus\{0\}$, when $U_\rmA$ is sampled from the Haar-random unitary ensemble, by virtue of \lref{lem:noisy_haar_incoherence}, we have
    \begin{equation}
        \<\mu_a|\mu_a\> = \frac{1}{d}\sum_{i,j=1}^\chi |\<l_j|\caN_\rmA(U_\rmA^\dag W_a U_\rmA)|l_i\>|^2 \le C_\rmA = O\left(\frac{\chi^2}{d^2}\operatorname{polylog}\frac{d}{\eta_1}\right).
    \end{equation}
    Hence, 
    \begin{equation}
        \|Z_s\| \le (d^2-1)C_\rmA + \|\bfP_\psi\|.
    \end{equation}
    Moreover, 
        \begin{align}
            \bbE[Z_s^2] &= (d^2-1)^2\left(\bbE[|\mu_a\>\<\mu_a|\mu_a\>\<\mu_a|] - \bbE[|\mu_a\>\<\mu_a|]^2\right) \nonumber \\
            &\le (d^2-1)^2 \left(C_\rmA \bbE[|\mu_{a_s}\>\<\mu_{a_s}|] - \bbE[|\mu_{a_s}\>\<\mu_{a_s}|]^2\right) \nonumber \\
            &= (d^2-1)C_\rmA(\bfP_\psi - |\mu_0\>\<\mu_0|) - (\bfP_\psi - |\mu_0\>\<\mu_0|)^2.
        \end{align}
    Hence, 
    \begin{equation}
        \left\|\sum_{s=1}^S \bbE[Z_s^2]\right\| \le (d^2-1) SC_{\rmA}\|\bfP_\psi\|.
    \end{equation}

    According to matrix Bernstein inequality in \thref{thm:singularv} with parameters
    \begin{gather}
        D \to \chi^2,\qquad \varepsilon \to \kappa \lambda_{\min}(\bfP_\psi)S, \\
        \mu \to (d^2-1)C_\rmA + \|\bfP_\psi\|, \qquad \nu \to (d^2-1)SC_\rmA\|\bfP_\psi\|,
    \end{gather}
    we have 
    \begin{equation}
        \Pr\left\{\left\|\sum_{s=1}^S Z_s\right\| \ge \kappa S\lambda_{\min}(\bfP_\psi)\right\} \le 2\chi^2 \exp\left[-\frac{\kappa^2 \lambda^2_{\min}(\bfP_\psi) S}{2(d^2-1)C_\rmA\|\bfP_\psi\|+ \frac{2}{3}[(d^2-1)C_\rmA + \|\bfP_\psi\|]\kappa \lambda_{\min}(\bfP_\psi)}\right].
    \end{equation}
    Therefore, if
    \begin{equation}
        S = O\left(\frac{ d^2C_\rmA\|\bfP_\psi\|}{\kappa^2 \lambda^2_{\min}(\bfP_\psi)}\log\frac{\chi^2}{\eta_2}\right),
    \end{equation}
    then with probability at least $1-\eta_2$,
    \begin{equation}
        \sigma_{\chi^2}^2(\Gamma_\rmA) \ge \sigma_{\chi^2}(\bbE[\Gamma_\rmA^\dag \Gamma_\rmA]) - \frac{1}{d^2-1}\left\|\sum_{s=1}^S Z_s\right\| \ge (1-\kappa)\frac{S}{d^2-1}\lambda_{\min}(\bfP_\psi).
    \end{equation}
    The theorem is proved by choosing $\eta_1 = \eta_2 = \eta/2$.
\end{proof}

The parameters $\lambda_{\min}(\bfP_\psi), \|\bfP_\psi\| = \lambda_{\max}(\bfP_\psi)$ can be bounded by the following lemma.
\begin{lemma}
\label{lem:noise-Gram-spectrum}
Suppose the Pauli channel $\caN_\rmA$ satisfies
    $\sum_{\alpha\neq0}x_\alpha
    \le
    \epsilon_{\rm Pauli} < 1/2$,
then
\begin{equation}
    (1-2\epsilon_{\rm Pauli})^2
    \le
    \lambda_{\min}(\bfP_\psi)
    \le
    \lambda_{\max}(\bfP_\psi)
    \le
    1.
\end{equation}
\end{lemma}

\begin{proof}
Introduce vector
\begin{equation}
    |z\rangle
    =
    \sum_{i,j=1}^{\chi}
    z_{ij}|i,j\rangle,
    \qquad
    \|z\|_2=1,
\end{equation}
define
\begin{equation}
    X_z
    \coloneqq
    \sum_{i,j=1}^{\chi}
    z_{ij}^*
    |l_i\rangle\langle l_j|.
\end{equation}
Since
$\{|l_i\rangle\langle l_j|\}_{i,j=1}^{\chi}$
is orthonormal with respect to the Hilbert--Schmidt inner product,
\begin{equation}
    \|X_z\|_F^2
    =
    \sum_{i,j=1}^{\chi}|z_{ij}|^2
    =
    1.
\end{equation}
By the definition of $\bfP_\psi$,
\begin{align}
    \langle z|\bfP_\psi|z\rangle
    =
    \sum_{i_1,j_1,i_2,j_2}
    z_{i_1j_1}^*
    z_{i_2j_2}
    \tr\left[
    \caN_{\rm A}(|l_{i_1}\rangle\langle l_{j_1}|)
    \caN_{\rm A}(|l_{i_2}\rangle\langle l_{j_2}|)^\dagger
    \right]=
    \|\caN_{\rm A}(X_z)\|_F^2.
\end{align}
Hence $\bfP_\psi \ge 0$. Let
\begin{equation}
    \caK_\psi
    \coloneqq
    \operatorname{span}
    \left\{
    |l_i\rangle\langle l_j|
    :
    1\le i,j\le\chi
    \right\}.
\end{equation}
Then
\begin{equation}
    \lambda_{\min}(\bfP_\psi)
    =
    \inf_{\substack{
    X\in\caK_\psi\\
    \|X\|_F=1
    }}
    \|\caN_{\rm A}(X)\|_F^2,\qquad  \lambda_{\max}(\bfP_\psi)
    =
    \sup_{\substack{
    X\in\caK_\psi\\
    \|X\|_F=1
    }}
    \|\caN_{\rm A}(X)\|_F^2.
\end{equation}
Suppose operator $X$ has the expansion
\begin{equation}
    X
    =
    \frac{1}{d}
    \sum_{a=0}^{d^2-1}
    \tr(W_a X)W_a.
\end{equation}
Let $\xi_a \coloneqq \sum_{\alpha}x_\alpha (-1)^{\gamma(a,\alpha)}$, then
\begin{equation}
    \caN_{\rm A}(X)
    =
    \frac{1}{d}
    \sum_{a=0}^{d^2-1}
    \xi_a\tr(W_a X)W_a,\qquad 
    \|\caN_{\rm A}(X)\|_F^2
    =
    \frac{1}{d}
    \sum_{a=0}^{d^2-1}
    \xi_a^2
    |\tr(W_aX)|^2,
\end{equation}
Note that
\begin{equation}
    \|X\|_F^2
    =
    \frac{1}{d}
    \sum_{a=0}^{d^2-1}
    |\tr(W_aX)|^2.
\end{equation}
It follows that
\begin{equation}
    \xi_{\min}^2\|X\|_F^2
    \le
    \|\caN_{\rm A}(X)\|_F^2
    \le
    \xi_{\max}^2\|X\|_F^2.
\end{equation}
Taking the infimum and supremum over normalized
$X\in\caK_\psi$ yields
\begin{equation}
    \xi_{\min}^2
    \le
    \lambda_{\min}(\bfP_\psi)
    \le
    \lambda_{\max}(\bfP_\psi)
    \le
    \xi_{\max}^2.
\end{equation}
Since $|\xi_a|
    \le
    \sum_\alpha x_\alpha
    =
    1,$
we have $\xi_{\max}\le1$, and therefore
    $\|\bfP_\psi\|
    =
    \lambda_{\max}(\bfP_\psi)
    \le
    1.$

On the other hand, because
   $ \sum_{\alpha\neq0}x_\alpha
    \le
    \epsilon_{\rm Pauli}$, we have $x_0\ge1-\epsilon_{\rm Pauli}$, and for every $a$,
\begin{align}
    \xi_a
    &=
    x_0
    +
    \sum_{\alpha\neq0}
    x_\alpha(-1)^{\gamma(a,\alpha)}
    \ge
    x_0
    -
    \sum_{\alpha\neq0}x_\alpha
    \ge
    1-2\epsilon_{\rm Pauli}.
\end{align}
Therefore, for $\epsilon_{\rm Pauli}<1/2$,
    $\xi_{\min}
    \ge
    1-2\epsilon_{\rm Pauli}$,
which implies
    $\lambda_{\min}(\bfP_\psi)
    \ge
    (1-2\epsilon_{\rm Pauli})^2$.
\end{proof}

Combining \lref{lem:noise-Gram-spectrum} and \thref{thm:pauli_singular_value}, we obtain the following corollary.
\begin{corollary}
    Suppose $U_\rmA, U_\rmB$ are drawn from the Haar random unitary ensemble, and $\caA = (a_1,a_2,\ldots, a_S)$ is a sequence of random binary strings with each $a_s$ drawn uniformly i.i.d. from $\bbF_2^{2n}\setminus\{0\}$. If the Pauli channels $\caN_\rmA, \caN_\rmB$ satisfy $\sum_{\alpha\neq 0}x_\alpha < \epsilon_{\rm Pauli} < 1/2, \sum_{\beta\neq 0}y_\beta < \epsilon_{\rm Pauli} < 1/2$, and 
    \begin{equation}
        S = O\left(\frac{\chi^2}{\kappa^2(1-2\epsilon_{\rm Pauli})^4}\log \frac{d}{\eta}\log\frac{\chi}{\eta}\right)
    \end{equation}
    Then with probability at least $1-\eta$,
    \begin{equation}
        \operatorname{rank}(\bfG_{\rm uc,\caA}) = \chi^2,\qquad \sigma_{\chi^2}(\bfG_{\rm uc,\caA}) \ge (1-\kappa)\frac{dS}{d^2-1}(1-2\epsilon_{\rm Pauli})^2 \lambda_{\rm lb}.
    \end{equation}
\end{corollary}

Since the lower bound on the singular value is changed, we need to adapt the measurement shot budget to
\begin{equation}
    M = \Theta\left(\frac{1}{(1-2\epsilon_{\rm Pauli})^4}\frac{d^2\log S}{S\lambda_{\rm lb}^2}\right).
\end{equation}

\subsubsection{Depolarizing noise}

Now we consider the depolarizing noise occurring on the entire system:
\begin{equation}
    \rho_{\mathrm{dp}} = (1-\epsilon)|\psi\>\<\psi| + \epsilon\frac{I}{d^2}.
\end{equation}
By linearity, the new CM (denoted by $\bfG_{\mathrm{dp}}$) is simply the linear combination of two other CMs:
\begin{equation}
    \bfG_{\mathrm{dp}} = (1-\epsilon)\bfG_\psi + \epsilon \bfG_0,\quad \bfG_0 \coloneqq |0\>\<0|.
\end{equation}
Accordingly, in the projected basis, we have
\begin{equation}
    \bfG_{\mathrm{dp},\caA} = (1-\epsilon)\bfG_{\psi,\caA} + \epsilon |0\>\<0|.
\end{equation}
Hence,
\begin{equation}
    \mathrm{rank} (\bfG_{\psi,\caA}) - 1 \le \mathrm{rank}(\bfG_{\mathrm{dp},\caA}) \le \mathrm{rank} (\bfG_{\psi,\caA}) + 1.
\end{equation}
If $\mathrm{rank}(\bfG_{\mathrm{dp},\caA}) > m^2+1$, we can still safely conclude that $\mathrm{rank} (\bfG_{\psi,\caA}) > m^2$, and hence $\sn(|\psi\>) > m$. In such a case, our protocol is robust against depolarizing noise. However, if $\sn(|\psi\>) = m, \mathrm{rank}(\bfG_{\psi}) = m^2$ and $\mathrm{rank}(\bfG_{\mathrm{dp}}) = m^2+1$, the extra rank 1 falsely gives us the conclusion that $\sn(|\psi\>) \ge m + 1$. Therefore, in the presence of depolarizing noise, when the number of dominant singular values is $m^2+1$, we can only safely conclude that $\sn(|\psi\>) \ge m$ instead of $m+1$.

 Note that in the worst case where $\epsilon = 1$, the state becomes maximally mixed and no information of $\sn(|\psi\>)$ can be extracted. To avoid this, we must have an upper bound on the noise size: $\epsilon \le \epsilon_{\rm ub} \in [0,1)$. Since $\bfG_0$ only has rank 1, the interlacing inequality gives
\begin{equation}
    \sigma_{j}(\bfG_{\mathrm{dp},\caA}) \ge (1-\epsilon)\sigma_{j+1}(\bfG_{\psi,\caA}) \ge (1-\epsilon_{\rm ub})\frac{(1-\kappa)dS}{d^2-1}\lambda_{\rm lb}\quad j = 1,\ldots, \chi^2-1,
\end{equation}
and
\begin{equation}
    \sigma_j(\hat{\bfG}_{\rm dp,\caA}) \ge (1-\epsilon_{\rm ub})(1-\kappa)\frac{dS}{d^2-1}\lambda_{\rm lb} - \sigma_{\rm th}(\eta).
\end{equation}
Hence, we need to increase the measurement shots to
\begin{equation}
    M = \Theta\left(\frac{1}{(1-\epsilon_{\rm ub})^2}\frac{d^2\log S}{S\lambda_{\rm lb}^2}\right).
\end{equation}

\section{A generalized OTOC interpretation of the incoherence parameter}

The incoherence parameter $\mu(U)$ is intimately linked to operator spreading \cite{adam2018operator}. A standard diagnostic for this phenomenon is the \textit{out-of-time-ordered correlator} (OTOC), defined for two unitaries $W$ and $V(t) = e^{\rmi Ht}V e^{-\rmi Ht}$ as $C(W,V,t) \coloneqq \text{Tr}\{[W,V(t)]^\dag [W,V(t)]\}/d$.  Expanding the commutator reveals its relationship to a four-point correlation function:
 \begin{equation}
      C(W,V,t) = 2 - \frac{2}{d}\text{Re}\{\text{Tr}[W^\dag V(t)^\dag WV(t)]\}.
 \end{equation}
 
 In chaotic systems, the constituent four-point functions typically decay, causing the OTOC to saturate.

In our framework, we take $W$ to be the projector $\Pi_l \coloneqq \sum_{i=1}^{\chi}|l_i\rangle\langle l_i|$ and $V$ as a Pauli operator $P$. This yields a generalized square commutator:
\begin{equation}
\mu(U(t)) = d \chi\left[ 1 - \frac{1}{2}\min_{a\neq 0} C(\Pi_l,P,t) \right],\qquad C(\Pi_l, P,t) = 2 - \frac{2}{\chi}\text{Tr}[\Pi_l P(t)\Pi_l P(t)].
\end{equation}
While $\Pi_l$ is a projector rather than a unitary, the physical intuition remains: both quantities measure the overlap between a time-evolved operator $P(t)$ and a reference operator. Notably, while standard OTOCs usually probe scrambling in physical space, $\mu(U(t))$ characterizes incoherence relative to the Schmidt vector basis. When $\mu(U(t))\approx 0$, every $a\neq 0$ is nearly saturated: $C(\Pi_l,P,t) \approx 2$. 

Unlike a conventional OTOC involving a fixed pair of local operators,
\(\mu(U(t))\) imposes a worst-case condition over all
\(P\in\mathcal{P}_0\). It therefore characterizes uniform spreading relative to
the subspace spanned by Schmidt vectors, rather than spatial operator growth alone. Moreover,
because \(\Pi_l\) depends only on the subspace
$\operatorname{span}(\{|l_i\>\}_{i=1}^{\chi})$, this interpretation is
invariant under a change of basis within that subspace.

\end{document}